\documentclass[11pt]{article}
\usepackage[left=1in,right=1in,top=1in,bottom=1in]{geometry} 
\usepackage{graphicx}
\usepackage{amssymb}
\usepackage{amsthm}
\usepackage[round,sort]{natbib}
\usepackage{subcaption}
\usepackage{amsmath,comment,color}
\graphicspath{{Images/}}
\usepackage{bm} 
\usepackage[noabbrev]{cleveref} 
\usepackage{authblk}

\newcommand{\bu}{{\bf u}}

\newcommand{\bB}{{\bf B}}

\newcommand{\rhat}{\hat{\bf r}}

\newcommand{\bn}{{\boldsymbol \nabla}}

\newcommand{\shat}{\hat{\bf s}}
\newcommand{\zhat}{\hat{\bf z}}
\newcommand{\phihat}{{\hat{\boldsymbol \phi}}}

\newcommand{\curl}{\bn\times}

\newcommand{\bBdot}{{\dot{\bB}}}
\newcommand{\proj}[1]{\overline{#1}}

\newcommand{\pd}[2]{\dfrac{\partial #1}{\partial #2}} 
\renewcommand{\vec}[1]{\bm{#1}} 
\newcommand{\grad}{\vec{\nabla}}
\newcommand{\unitr}{\hat{\vec{r}}}
\newcommand{\laplacian}{{\nabla^2}}

\newtheorem*{theorem}{Theorem}

\title{Three-dimensional solutions for the geostrophic flow in the Earth's core}
\author[1]{Colin M. Hardy}
\author[2]{Philip W. Livermore}
\author[3]{Jitse Niesen}
\author[4]{Jiawen Luo}
\author[4]{Kuan Li}
\affil[1]{EPSRC Centre for Doctoral Training in Fluid Dynamics, University of Leeds, Leeds, LS2 9JT, UK}
\affil[2]{School of Earth and Environment, University of Leeds, Leeds, LS2 9JT, UK}
\affil[3]{School of Mathematics, University of Leeds, Leeds, LS2 9JT, UK}
\affil[4]{Institut f\"ur Geophysik, ETH Zurich, Sonneggstrasse 5, 8092 Z\"urich, Switzerland}

\begin{document}
\maketitle

{\color{red}
}

\abstract
In his seminal work, \cite{Taylor_63} argued that the geophysically relevant limit for dynamo action within the outer core is one of negligibly small inertia and viscosity in the magnetohydrodynamic equations. Within this limit, he showed the existence of a necessary condition, now well known as Taylor's constraint, which requires that the cylindrically-averaged Lorentz torque must everywhere vanish; magnetic fields that satisfy this condition are termed Taylor states. Taylor further showed that the requirement of this constraint being continuously satisfied through time prescribes the evolution of the geostrophic flow, the cylindrically-averaged azimuthal flow. 
We show that Taylor's original prescription for the geostrophic flow, as satisfying a given second order ordinary differential equation, is only valid for a small subset of Taylor states. An incomplete treatment of the boundary conditions renders his equation generally incorrect. 

Here, by taking proper account of the boundaries, we describe a generalisation of Taylor's method that enables correct evaluation of the instantaneous geostrophic flow for any 3D Taylor state. We present the first full-sphere examples of geostrophic flows driven by non-axisymmetric Taylor states. Although in axisymmetry the geostrophic flow admits a mild logarithmic singularity on the rotation axis, in the fully 3D case we show that this is absent and indeed the geostrophic flow appears to be everywhere regular.



\section{Introduction}

Earth's magnetic field is generated by a self-excited dynamo process through the flow of electrically-conducting fluid in the outer core. Although the set of equations that govern this process are known, their numerical solution is challenging because of the extreme dynamical conditions \citep{Roberts_King_2013}.   
Of particular note is the extreme smallness of the core's estimated viscosity, and the large disparity between the daily timescale associated with Earth's rotation and the thousand-year timescale that governs the long-term geomagnetic evolution. Represented in terms of non-dimensional numbers, this means that
%
the Rossby number (also known as the magnetic Ekman number, $E_\eta$, measuring the ratio of rotational to magnetic timescales) is ${R_o = O(10^{-9}})$ and the Ekman number (measuring the ratio of rotational to viscous effects) is ${E=O(10^{-15})}$. The smallness of these parameters means that rapid (sub-year) timescales associated with inertial effects (e.g. torsional waves) and extremely thin boundary layers (of depth about 1~m) must be resolved in any Earth-like numerical model, even though neither likely plays an important role in the long term evolution of the geodynamo. 

Over the past decades, modellers of the long term geomagnetic field have followed one of two largely independent strategies in order to circumvent these problems.
First, beginning with the work of \cite{Glatzmaier_Roberts_95a} and \cite{Kageyama_etal_95}, it was noted that by artificially increasing these two parameters by many orders of magnitude to now typical values of ${R_o = 10^{-3}}$, ${E = 10^{-7}}$ \citep{Christensen_2015}, the numerically difficult rapid timescales and short length scales are smoothed, allowing larger time steps, and therefore ultimately permitting a longer time period to be studied for a given finite computer resource. Although such (now mainstream) models can reproduce many characteristics of Earth's geomagnetic field, several studies have cast doubt as to whether they obey the correct force balance within the core \citep{soderlund2012influence,king2013flow,Roberts_King_2013}, although some evidence points to models being on the cusp of faithfully representing Earth's core \citep{yadav2016approaching,schaeffer2017turbulent}.


In the second strategy, which we consider here in this paper, the values of $R_o$ and $E$ are both set to zero \citep{Taylor_63}. By entirely neglecting inertia and viscosity, the challenging aspects of rapid timescales and very short viscous lengthscales are removed and this approximation will likely lead to a computationally less demanding set of equations to solve. The resulting dimensionless magnetostrophic regime then involves an exact balance between the Coriolis force, pressure, buoyancy and the Lorentz force associated with the magnetic field $\bB$ itself:
\begin{equation}\vec{\hat z} \times \vec{u} = -\grad p + F_B\vec{\hat r} + \curl \vec{B} \times \vec{B}, \label{eqn:magneto}
\end{equation}
where $F_B$ is a buoyancy term that acts in the unit radial direction $\vec{\hat r}$ and $\vec{\hat z}$ is the unit vector parallel to the rotation axis \citep{Fearn_98}. 
In a full sphere (neglecting the solid inner core), 
a complete description of the geodynamo requires a solution of \eqref{eqn:magneto} alongside equations describing the evolution of $\bf B$ and $F_B$ within the core, whose boundary conditions derive from the surrounding electrically-insulating impenetrable overlying mantle.
Denoting $(s,\phi,z)$ as cylindrical coordinates, \citet{Taylor_63} showed that, as a consequence of this magnetostrophic balance, the magnetic field must obey at all times $t$ the well-known condition
\begin{equation} 
T(s,t) \equiv \int_{C(s)} ([\curl \bB] \times \bB)_\phi s d\phi dz =0,\label{eqn:Taylor} 
\end{equation}
for any geostrophic cylinder $C(s)$ of radius $s$  coaxial with the rotation axis.

Taylor also showed that it is expedient to partition the magnetostrophic flow of \eqref{eqn:magneto}, using a cylindrical average, into geostrophic and ageostrophic parts:
$$ {\bf u} =  u_g(s) \phihat + {\bf u}_{a}(s,\phi,z),$$
in which the ageostrophic flow ${\bf u}_{a}$ has an azimuthal component with zero cylindrical average. 
Provided equation \eqref{eqn:Taylor} is satisfied, equation \eqref{eqn:magneto} can be used to find ${\bf u}_{a}$ directly (for example, by using Taylor's constructive method or the integral method of \citet{Roberts_King_2013}), although the geostrophic flow remains formally unspecified by \eqref{eqn:magneto}. As Taylor further showed however, the geostrophic flow can be constrained by insisting that \cref{eqn:Taylor} is not just satisfied instantaneously but for all time. The task of $u_g$ is then to keep the magnetic field on the manifold of Taylor states \citep{Livermore_etal_2011}. It is noteworthy that, in such a model, at all times the flow is slaved to $\bf B$ and $F_B$.

In his 1963 paper, Taylor showed that (for a fully 3D system) the geostrophic flow was at every instant the solution of a certain second order differential equation (ODE) whose coefficients depend on $\bB$ and $F_B$. His elegant and succinct analysis has been reproduced many times in the literature. It may then come at some surprise that in the intervening five decades there have been no published implementations of his method (that the authors are aware of).
Very likely, this is due to a subtle issue concerning the treatment of the magnetic boundary conditions. As we shall show, rather than being applicable to a general (Taylor state) $\bB$, Taylor's method is only valid for a small subset of Taylor states. Of crucial importance is that this subset does not include those states likely to be realised in any analytical example or in any practical numerical scheme to solve the magnetostrophic equations. 
The main goal of this paper is to describe why this happens, and to modify Taylor's method in order that it can apply more generally. 



Despite the lack of headway using a direct application of Taylor's ODE, some alternative methods to evolve the magnetostrophic equation have shown success. 
By treating a version of the Taylor integral \eqref{eqn:Taylor} that is specific to axisymmetry \citep{Braginsky_70b, Jault_95}, \cite{Wu_Roberts_2015} demonstrated that they could evolve the magnetostrophic system by solving a first order differential equation for the geostrophic flow, rendering the Taylor integral zero to first order, and went on to apply it to a variety of examples.  
In an independent line of investigation \cite{Li_etal_2018} showed that, by using control theory, it is possible to find $u_g$ implicitly such that the Taylor integral is zero at the end of any finite timestep. As we show later explicitly by example, their method is fundamentally 3D,  although in their paper they only applied it to the axisymmetric case.
The generalised version of Taylor's method that we present in this paper is also fully 3D and provides an alternative means to that of \citet{Li_etal_2018} of calculating the geostrophic flow. Either of these methods may provide a route to create a fully 3D magnetostrophic alternative to the mainstream numerical models with weak viscosity and inertia. We note however, that the methods we describe within this paper are restricted to the full sphere, we do not attempt to incorporate the inner core or any of its dynamical effects.


An alternative route to finding a magnetostrophic dynamo is to reinstate viscosity and/or inertia and investigate the limit as both $E$ and $R_o$ become small \citep{Jault_95}. Arguably this would result in models closer to geophysical reality than those that are purely magnetostrophic as this is precisely the regime of the Earth's core. A variety of studies reported evidence of behaviour independent of $E$ in the inviscid Taylor-state limit, either from a direct solution \citep{Fearn_Rahman_2004}, or from solving the equations assuming asymptotically small $E$ \citep{Hollerbach_Ierley_91, Soward_Jones_83}. To date, all models of this type have been axisymmetric and there have been no attempts at a general 3D implementation of these ideas. One difficulty with treating asymptotically-small $E$ is that the resulting equation for $u_g$ is an extremely delicate ratio of two small terms, whose form is dependent on the specific choice of mechanical boundary conditions \citep{Livermore_etal_2016a}. The convergence of magnetostrophic and asymptotically low-$E$ models remains an outstanding question.

The remainder of this paper is structured as follows.
Before we can explain why Taylor's method of determining the geostrophic flow fails in general, we need to set out some general background and review other alternative schemes: this is accomplished in sections 2--5. In section 6 we discuss the importance of a key boundary term and why it restricts the validity of Taylor's method; we then show explicitly in a simple case that Taylor's method fails. 
In sections 8--10 we generalise Taylor's method and give some examples, discussing the existence of weak singularities in section 11; we end with a discussion in section 12.

\section{General considerations}
\label{sec:general_considerations}

\subsection{Non-dimensionalisation}
\label{sec:scaling}


In the non-dimensionalisation considered in this paper, length is scaled by $L$, the outer core radius $3.5 \times 10^6$~m, time by $\tau$, the ohmic diffusion time $\tau$ (250--540~kyr) \citep{Davies_etal_2015}, and speed by ${\mathcal U} = L\tau^{-1}\approx 5 \times 10^{-7}$. The scale used for the magnetic field is ${\mathcal B} = (2\Omega_0 \mu_0 \rho_0 \eta)^\frac{1}{2}$ \citep{Fearn_98}, where for Earth the physical parameters take the following values: angular velocity $\Omega_0=7.272 \times 10^{-5}$~s, permeability $\mu_0=4\pi \times 10^{-7}~NA^{-2}$, density $\rho_0=10^{4}~$kg m$^{-3}$ and magnetic diffusivity $\eta=0.6$--$1.6$~m$^2$s$^{-1}$. These parameters lead to the non-dimensional parameters $R_o = \eta/(2\Omega L^2)\approx 10^{-9}$ and $E = \nu/(2\Omega L^2) \approx 10^{-15}$, whose small values motivate neglecting the terms they multiply.

The value of ${\mathcal B} \approx 1.7$~mT is close to the estimate of the geomagnetic field strength of \citet{Gillet_etal_2010}, and so we use dimensionless magnetic fields with toroidal or poloidal components of rms (root mean squared) strength of unity. This corresponds to a dimensional rms magnitude of 1.7~mT for purely toroidal or purely poloidal fields and $1.7\sqrt{2} \approx 2.4$~mT for mixed states. Using $\mathcal U$, this choice enables the immediate interpretation of the dimensional scale of any flow that we show.


\subsection{Magnetic field representation and the initial state}
In our full sphere of unit radius, the position $\vec{r}$ is naturally described in spherical coordinates  $(r,\theta,\phi)$, although the importance of the rotation axis also leads us to use cylindrical coordinates $(s,\phi,z)$. The magnetic field ${\bf B}$ can be written using a toroidal (T)-poloidal (S) framework
$$ \bB = \curl\curl \mathcal{S}\vec{\hat r} + \curl \mathcal{T} \vec{\hat r},$$
with $S$ and $T$ expanded as
$$ \mathcal{S} = \sum_{l,m} \mathcal{S}_l^m(r) Y_l^m(\theta,\phi), \qquad \mathcal{T} = \sum_{l,m} \mathcal{T}_l^m(r) Y_l^m(\theta,\phi),$$
where $Y_l^m$ is a spherical harmonic of degree $l$ and order $m$. The functions $\mathcal{S}$ and $\mathcal{T}$ 
must be chosen to satisfy both Taylor's condition \eqref{eqn:Taylor}, along with the electrically insulating boundary conditions at $r=1$ that
can be written
\begin{equation} \frac{d \mathcal{S}_l^m}{dr} + l \mathcal{S}_l^m = \mathcal{T}_l^m = 0. \label{eqn:bc}
\end{equation}
The fluid is assumed to be incompressible and hence the flow $\bu$ can also be written in a comparable form, and due to the absence of viscosity only satisfies an impenetrability condition: $u_r = 0$ on $r=1$. We cannot impose no-slip or stress-free conditions, there being no boundary layer to accommodate any adjustment from the free-stream inviscid structure. 

All time-dependent magnetostrophic models, axisymmetric or 3D, require an initial state from which the system evolves. Because the flow is defined completely by the magnetic field and $F_B$, only the initial structure of the magnetic field ${\bf B}(0)$ and $F_B(0)$ are needed: there is no need to specify the initial flow. 
%
A general scheme for finding an exact initial Taylor state using a poloidal-toroidal representation was described in \citet{Livermore_etal_2008}; in general it requires a highly specialised magnetic field to render its integrated azimuthal Lorentz force zero over all geostrophic cylinders. However, in a full sphere such cancellation can be achieved in a simple way by exploiting reflectional symmetry in the equator \citep{Livermore_etal_2009a}. Using the Galerkin basis of single-spherical-harmonic modes that satisfy the boundary conditions (see \Cref{sec:Apb}), suitable simple modal expansions are automatically Taylor states. 

\subsection{Overview of time evolution}

Because of the absence of inertia, at each instant the magnetostrophic flow is entirely determined by $\bf B$ and $F_B$ from \cref{eqn:magneto}: therefore the system, as a whole, only evolves through time-evolution of the quantities $F_B$ and $\bf B$. The evolution of $F_B$ is assumed to be tractable and lies outside the scope of this study: for simplicity we shall henceforth assume that $F_B = 0$, although we note that all the methods nevertheless apply in the case of non-zero $F_B$. The evolution of the magnetic field is described by the induction equation:
\begin{equation} \label{eqn:induct} \partial_t {\bB}(\vec{r},t) = {\mathcal I}(\bB, \bu) \equiv \curl\big[\bu\times \bB(\vec{r},t) \big] + \eta \nabla^2 \bB(\vec{r},t)
\end{equation}
where $\eta\ne 0$ is the magnetic diffusivity (assumed constant) and $\partial_t = \partial/\partial t$.
Assuming that we can evolve $\bB$ and $F_B$ (using standard methods), the major outstanding task is then to determine the flow at any instant given $\bf B$ and $F_B$. 

The ageostrophic component of the flow, containing all the (possibly complex) axially asymmetric structure turns out to be straight-forward to calculate, as it can be determined either through the integral method of \citet{Roberts_King_2013}, the constructive method of \citet{Taylor_63} or the spectral method as described in \Cref{sec:u_a_method}. By contrast, the more elementary  geostrophic flow, depending only on $s$, is surprisingly difficult to compute, owing to its key role of maintaining Taylor's constraint. 

There are two ways in which the geostrophic flow may be found, which differ in philosophy. In the first, we may undertake an instantaneous analysis to find
the geostrophic flow that gives zero rate of change of Taylor's constraint: $\partial_t T(s,t) = 0$ \citep{Taylor_63}. Because of the resulting closed-form analytic description, such methods can be useful in computing snapshot solutions that elucidate the mathematical structure of the geostrophic flow, for example, the presence of any singularities.
However, as a practical time-evolution tool, their utility is not so obvious. For example, the simple explicit time-evolution scheme, defined by assuming an instantaneous solution is constant over a finite time interval, would lead to a rapid divergence from the Taylor manifold \citep[see][for an example]{Livermore_etal_2011}. 

In the second type of method, we may consider taking a time step (of size $h$), determining the geostrophic flow implicitly by the condition that the magnetic field $\bB(t+h)$ satisfies Taylor's constraint \citep{Li_etal_2018,Wu_Roberts_2015}. 
%
%
In general, implicit and instantaneous methods methods will only produce the same geostrophic flow in a steady state, or for a time-dependent state for infinitesimally small $h$. 


All methods to determine the geostrophic flow do so up to an arbitrary solid body rotation: $u_g=as$. The constant $a$ can be found through requiring zero global angular momentum
\begin{equation} \label{eqn:angmomcon} \int_0^1\int_{-Z_T}^{Z_T}\int_0^{2\pi} s (\vec{u}_{a} \cdot \phihat+u_g) ~d\phi dz s ds = 0, \end{equation}
where $Z_T = \sqrt{1-s^2}$ is the half-height of $C(s)$. We also assume the geostrophic flow is everywhere finite, which is implemented by additional conditions where necessary.

\section{Braginsky's formulation}
\label{sec:Brag}
Before discussing the determination of the geostrophic flow in more detail, we briefly review a crucial alternative formulation of Taylor's constraint due to \cite{Braginsky_70b}, which laid the foundations of many subsequent works on the subject \citep[e.g.][]{Roberts_Aurnou_2011,Wu_Roberts_2015,Braginsky_75,fearn1992magnetostrophic,Book_Jault_2003}. As an identity the Taylor integral \eqref{eqn:Taylor}
can be equivalently written
\begin{equation} T(s,t) =  \frac{1}{s} \frac{\partial}{\partial s} \bigg[ s^2 \int_{C(s)} B_\phi B_s d\phi dz \bigg] + \frac{s}{\sqrt{1-s^2}} \oint_{N+S} (B_\phi B_r) d\phi,\label{eqn:manipulation1} \end{equation}
where $N$ and $S$ are the northern and southern end caps of the cylinder $C(s)$ at the intersection with the spherical boundary at $r=1$. 

It is also useful to consider the net magnetic torque on all fluid enclosed within $C(s)$, $\Gamma_z$, defined by 
$$T(s,t) = \frac{1}{s}\, \frac{\partial \Gamma_z}{\partial s} \qquad\text{or} \qquad \Gamma_z(s,t) = \int_0^s s' T(s',t) ds'. $$
In our full-sphere geometry, it is clear that $\Gamma_z(s,t)$ is zero if and only if $T(s,t)$ is zero, although in a spherical shell it is possible that a piecewise (non-zero) solution exists for $\Gamma_z$. 
The condition $\Gamma_z=0$ defines what we refer to as the Braginsky constraint:
\begin{equation} 0=\Gamma_z \equiv  s^2 \int_{C(s)} B_\phi B_s d\phi dz  + \int_0^s \oint_{N+S} \frac{{s'}^2\,B_\phi B_r}{\sqrt{1-{s'}^2}}  d\phi\,ds',\label{eqn:Taylor_alternative} \end{equation}
which is equivalent to Taylor's constraint, and simplifies for specific classes of magnetic fields that cause the boundary term to vanish. 
One such class relates to magnetic fields with no radial component on $r=1$ (e.g. toroidal fields), a
further class is that whose fields have a vanishing azimuthal component on $r=1$ (e.g. axisymmetric fields).

It is important to note the significant difference in the mathematical structure between the constraints of Braginsky \eqref{eqn:Taylor_alternative} and Taylor \eqref{eqn:Taylor}. In \eqref{eqn:Taylor_alternative} there is a clear partition between the two surface integral terms on the right hand side: the first term is an integral defined over $C(s)$ that is independent of the magnetic field values on the end caps 
(these being a set of measure zero); the second end-cap term depends only on the boundary values of the magnetic field. By contrast, although ostensibly Taylor's integral \eqref{eqn:Taylor} is an integral over the surface $C(s)$, the integrand involves a spatial derivative (the curl of $\bf B$) leading to a dependence on the boundary values of the magnetic field. As we will see later, this hidden dependence on the boundary conditions has a deep consequence on Taylor's method for determining the geostrophic flow.

\section{Existing methods to determine the geostrophic flow}
Our modification of Taylor's method described in \cref{sec:Altmeth} determines the instantaneous geostrophic flow in a fully 3D geometry. In this section, we briefly review other methods available to calculate the geostrophic flow whose working assumptions are different: either they are axisymmetric, or designed to take a finite time step and are not instantaneous. Where there is overlap in applicability, we will use these methods to numerically confirm our solutions.

\subsection{An axisymmetric first-order implicit method}

As noted above, under axisymmetry Braginsky's condition collapses to 

\begin{equation} \Gamma_z = 2\pi s^2 \int_{-Z_T}^{Z_T} B_\phi B_s dz =0. \label{eqn:TC_axi} \end{equation}
This simple form was exploited by \citet{Wu_Roberts_2015} who considered taking a single timestep of duration $h$, after which they required 
\begin{equation} 
 \Gamma_z(s,t)+h\pd{\Gamma_z(s,t)}{t} =0. \label{eqn:first_order}
\end{equation}
The left hand side here approximates $\Gamma_z(s,t+h)$, so this ensures that \eqref{eqn:TC_axi} is satisfied to first order.
To find an equation for the geostrophic flow they differentiated \cref{eqn:TC_axi} with respect to time and used the fact that the geostrophic term in the induction equation reduces to
\begin{equation} \label{eqn:simpinduct} \curl(u_g(s) \phihat  \times\bB) = s B_s \frac{\partial (u_g/s)}{\partial s} \phihat.
\end{equation}
They obtained the following first order ordinary differential equation describing the geostrophic flow
\begin{equation} s \alpha_0(s)   \frac{d}{ds}\left(\frac{u_g(s)}{s} \right) = - S_0(s) - \frac{\Gamma_z(s,t)}{h}, \label{eqn:1ODE} \end{equation}
where \[  S_0(s) = 2\pi s^2 \int_{-Z_T}^{Z_T} ( B_s {C}^{a}_\phi + B_\phi {C}^{a}_s)\, dz,\qquad \alpha_0(s) = 2\pi s^2 \int_{-Z_T}^{Z_T} B_s^2\, dz,\]
and 
\begin{equation} \vec{C}^{a}=\curl(\bu_{a} \times \bB) + \eta \nabla^2 \bB. \label{eqn:Cdef} \end{equation} 
The subscripts of zero denote a restriction to axisymmetry of (more general) 3D quantities that are defined subsequently. \citet{Wu_Roberts_2015} implemented this method by solving \cref{eqn:1ODE} using a finite difference scheme. It is worth remarking that this scheme allows small numerical deviations from a Taylor state (since \eqref{eqn:first_order} is only approximate).
Because the method depends upon \eqref{eqn:TC_axi} which is tied to axisymmetry, their method is not extendable to 3D.

\subsection{A 3D fully implicit scheme} \label{sec:3D_implicit}
An alternative implicit scheme proposed by \citet{Li_etal_2018}, was to seek a geostrophic flow that ensured Taylor's constraint is satisfied (without error) in a numerical scheme after taking a single timestep $h$. By extending to multiple timesteps, this method is suitable to describe fully 3D time-dependent dynamics. Although the authors only demonstrated its utility on axisymmetric examples, in this paper we will show how the method simply extends to 3D with a single short time-step.

The key idea is to minimise (hopefully to zero) the target function
\begin{equation} \Phi = \int_0^1 T^2(s,t+h) ds \label{eqn:Li_1} \end{equation}
by optimising over all possible choices of $u_g$, assumed constant throughout the interval $0 \le t \le h$. Although \cite{Li_etal_2018} set out a sophisticated algorithm to do this in general based on control theory, here we describe a simplification of the method which is suitable for $h \ll 1$, which we can use to benchmark our instantaneous solutions of the generalised 3D Taylor methodology.

Like \cite{Li_etal_2018} we adopt a modal expansion of $u_g$, of which a general form is 
\begin{equation} 
\label{eqn:specexp} u_g = s\sum_{i=0}^I A_i T_i(2s^2-1)+Bs\ln(s) \end{equation}
where $T_i(2s^2-1)$ are even Chebyshev polynomials of the first kind, and we allow a weak logarithmic singularity at the origin as required by our analytic results in \cref{sec:Tayfail}; see also \cref{sec:sing}. 

Because we plan to take only a single time step of size $h\ll 1$, we adopt a very simple first order explicit Euler time evolution scheme 
$$\bB(t+h)=\bB(t)+h\, \partial_t \bB(t)$$
which is then substituted into \eqref{eqn:Li_1}. For simplicity we assume that the ageostrophic flow, calculated at $t=0$, is also constant over the time-step. 
As a representation of the magnetic field (and its rate of change), we use a Galerkin scheme (see \Cref{sec:Apb}), which satisfies the boundary conditions \eqref{eqn:bc} automatically. Practically, this means that we use $\proj{\mathcal I}$ (see \cref{eqn:induct}) in place of $\partial_t \bB$, where the overbar denotes the projection onto the Galerkin basis. The coefficients $A_i$ and $B$ are then found through minimising $\Phi$. We note that since $\bB(t+h)$ is formally linear in $u_g(s)$, $T(s,t+h)$ is then quadratic and hence $\Phi$ quartic in the coefficients $A_i$ and $B$. \citet{Li_etal_2018} found the minimum using an iterative scheme, although we note that, in general (and without a good starting approximation), finding such a minimum may be problematic.
%

It is noteworthy, however, that in the axisymmetric case this analysis is greatly simplified. Through \cref{eqn:simpinduct} only the azimuthal component of $\bB(t+h)$ depends on $u_g$, and \cref{eqn:manipulation1} shows that $T(s)$ is now linear and $\Phi$ quadratic in $u_g$, hence finding the minimum of $\Phi$ is more straightforward.




\subsection{An instantaneous axisymmetric method}

\citet{Wu_Roberts_2015} also presented a method for finding an instantaneous solution for the geostrophic flow in axisymmetry.  Through differentiating with respect to time \cref{eqn:TC_axi} they arrive at the following first order ODE, here referred to as the BWR (Braginsky-Wu-Roberts) equation:
\begin{equation} {\cal L}_{BWR} \equiv s \alpha_0(s)   \frac{d}{ds}\left(\frac{u_g(s)}{s} \right) = - S_0(s), \label{eqn:Wuinstant} \end{equation}
which is the same as \eqref{eqn:1ODE} without the final term. This gives $u_g(s)$ explicitly as
\begin{equation} u_g(s) = -s \int_0^s \frac{S_0(s')}{s' \alpha_0(s')}~ \text{d}s'. \label{eqn:BWR_exact} \end{equation}
In all the cases we consider, \eqref{eqn:BWR_exact} can be solved analytically (with the assistance of computer algebra).
%
A further property of this equation is that, for a purely-poloidal axisymmetric magnetic field, the solution $u_g$ is independent of the magnetic diffusivity $\eta$. This is because $\laplacian \bB$ is also purely-poloidal and a purely-poloidal field has no azimuthal component. Thus 
$$ B_s(\laplacian \vec{B})_\phi = B_\phi (\laplacian \vec{B})_s = 0$$
and the diffusion term (within ${S}_0$) then never appears in \eqref{eqn:Wuinstant}.
This differs from the case of a more general magnetic field with both toroidal and poloidal components, where $u_g$ depends upon $\eta$.

We also observe that for an axisymmetric purely-toroidal field, since $B_s = 0$ everywhere \cref{eqn:Wuinstant} is null because $\alpha_0 = S_0=0$ reducing it to the tautology $0=0$ and hence placing no constraint on the geostrophic flow.




%
%
%
\subsection{Taylor's 3D instantaneous method}
\label{sec:Taylor_method}
We end this section by discussing the well known (instantaneous) method of Taylor, who determined the unknown geostrophic flow by differentiating with respect to time (denoted by the over-dot shorthand) the Taylor integral in \cref{eqn:Taylor} to produce:
\begin{equation} 0 = \int_{C(s)} \big\{ [\curl\bBdot]\times\bB  +  [\curl\bB]\times\bBdot \big\}_\phi\;s\,d\phi\,dz. \label{eqn:Taylor_ddt} \end{equation}
On substituting directly for $\bBdot$ from  \cref{eqn:induct}
in addition to its curl (describing $\curl \bBdot$), Taylor showed that for fully 3D Taylor states $\bB$ the resulting equation for the geostrophic flow can be written in a remarkably succinct form as the second order ordinary differential equation
\begin{equation}
{\cal L}_T(u_g) \equiv \alpha(s) \frac{d^2}{ds^2}\left(\frac{u_g(s)}{s} \right) + \beta(s) \frac{d}{ds}\left(\frac{u_g(s)}{s} \right) = G(s). \label{eqn:2ODE}
\end{equation}
In the above, the coefficients are
\begin{equation} \label{eqn:alpha}
\alpha(s) = \int_{C(s)} s^2\, B_s^2\, d\phi\, dz, \qquad
\beta(s) = \int_{C(s)} \left[ 2B_s^2 + s\, \bB\cdot \bn B_s\right]\, s\, d\phi\, dz, 
\end{equation}
and $G(s)$ is a function describing the interaction of $\vec{u}_{a}$ and the magnetic field defined as
$$G(s) =  - \frac{1}{s} \frac{\partial}{\partial s} \bigg[ s^2 \int_{C(s)} \vec{C}^{a}_\phi B_s + \vec{C}^{a}_s \, B_\phi d\phi dz \bigg].$$
Note the mistake in \cite{Taylor_63} where a factor of $s$ is omitted within the coefficient $\beta$.
The functions $\alpha_0$ and $S_0$, previously defined, are simply axisymmetric variants of $\alpha$ given above and $S(s)$ defined as
$$S(s) =   s^2 \int_{C(s)} ({C}^{a}_\phi B_s + {C}^{a}_s \, B_\phi) d\phi dz  + \int_0^s s' \bigg[ \frac{s'}{\sqrt{1-{s'}^2}} \oint_{N+S} ({B_\phi} {C}^{a}_r + {B_r} {C}^{a}_\phi )d\phi \bigg] ds',$$
where $\vec{C}^a$ is as defined in \cref{eqn:Cdef}.
The fact that the coefficients $\alpha(s)$ and $\beta(s)$ are spatially dependent means that analytic solutions to \eqref{eqn:2ODE} are very rare and in general only  numerical solutions are possible. 
Of crucial note is that the boundary conditions played no part in the derivation above. 

\section{Technical aside: higher order boundary conditions}
\label{sec:cty}

\subsection{Higher order boundary conditions in the heat equation}

Taylor's method is based on the instantaneous evolution (which we can take to be at time $t=0$) of the magnetostrophic system whose magnetic field is prescribed and must satisfy Taylor's constraint. Here we  discuss higher order boundary conditions, the importance of which has so far been overlooked. We start by introducing this concept in a simple PDE, then we discuss the relevance for Taylor's equation.

Suppose we are interested in finding $f(x,t)$ on $x \in [0,1]$, whose evolution is described by the heat equation in the interior of the domain
\[ \frac{\partial f}{\partial t} = \frac{\partial^2 f}{\partial x^2},\]
to be solved with the
boundary conditions $f(0,t) = f(1,t) = 0$. 
For this simple equation, the general solution can be written in the form
\[ f(x,t) = \sum_n A_n e^{-n^2 \pi^2\, t}\, \sin(n\,\pi\,x).\]
Let us now suppose we have an initial state:
\[ f(x,0) = x^2(1-x) \]
which satisfies the boundary conditions. Its future evolution would be given by the projection onto the normal modes as above.

In Taylor's analysis, part of the integral in \eqref{eqn:Taylor} could be converted to a boundary term. Here we consider an analogy which is exactly integrable:
\begin{equation} \frac{d}{dt} \int_0^1 \frac{\partial f}{\partial x} \, dx = \frac{d}{dt} [f(1)-f(0)] = 0  \label{eqn:simple_zero} \end{equation}
using the boundary conditions.
In Taylor's derivation, he differentiated under the integral sign and substituted directly for $\partial f/\partial t$, in order to find the equation that $u_g$ must satisfy using an instantaneous initial magnetic field. In our example, this produces
\begin{equation}
\label{taylor-wave}
\frac{d}{dt} \int_0^1 \frac{\partial f}{\partial x} \, dx = \int_0^1 \frac{\partial^2 f}{\partial x \partial t}\, dx = \int  \frac{\partial^3 f}{\partial x^3}dx = [f_{xx}(1,t)-f_{xx}(0,t)].
\end{equation}
At $t=0$, we evaluate the above expression as $-6$ (note that $f_{xxx}(x,0) = -6$)
resulting in an apparent contradiction with \eqref{eqn:simple_zero} and illustrating that this approach is not generally valid. 

The problem arises because the initial state does not satisfy the condition $f_{xx}(0,t) = f_{xx}(1,t) = 0$, which arises from differentiating $f(0,t) = f(1,t) = 0$ with respect to time and substituting the PDE. The condition $f_{xx}(0,t) = f_{xx}(1,t) = 0$ is called the first-order boundary condition \citep{Book_Evans_2010}. The consequence of the initial state not satisfying the first-order boundary condition is that the solution is not smooth at the boundary at $t=0$. Specifically, the derivatives in~\eqref{taylor-wave} do not exist and thus the above derivation is not valid. As a simple illustration of the issue, note that the general solution implies that $f_{xxx}(x,0) = -\sum_n n^3 \pi^3 A_n \cos(n \pi x)$, which cannot represent the constant function $f_{xxx}(x,0) = -6$ associated with the initial state. This lack of smoothness only occurs at the initial time $t=0$. At any later time ($t>0$), the solution is infinitely smooth; this is the smoothing property of the heat equation.

In the very special case that the initial state satisfies the first order boundary conditions (e.g. $f(x,0) = x^3(1-x)^3$) then there is no contradiction and \eqref{eqn:simple_zero} and \eqref{taylor-wave} are consistent. However, for a general initial condition, the procedure adopted is not valid. 


\subsection{The relevance for Taylor's equation}

We now discuss the relevance of the above discussion of higher-order boundary conditions in the context of the Earth's magnetic field.
In the derivation of Taylor's second-order ODE \eqref{eqn:2ODE}, it is implicitly assumed that $\bB$ and all its time derivatives are (initially) smooth everywhere. Although it is somewhat hidden in Taylor's original derivation, taking the time-derivative of the equivalent form of \eqref{eqn:manipulation1} makes this explicit:
\begin{equation} \frac{1}{s} \frac{\partial}{\partial s} \bigg[ s^2 \int_{C(s)} (\dot{B_\phi} B_s + B_\phi \dot{B_s}) d\phi dz \bigg] + \frac{s}{\sqrt{1-s^2}} \oint_{N+S} (\dot{B}_\phi B_r + B_\phi \dot{B_r}) d\phi=0.\label{eqn:alter1} \end{equation}
Taylor substituted everywhere the induction equation~\eqref{eqn:induct}, $\partial_t \bB = {\mathcal I}(\bu,\bB)$, but in view of the above discussion, we need to take care, particularly for the boundary terms.

We appeal to a reduced version of the magnetostrophic equations in order to probe what can be said about the behaviour of $\bB(t)$ on the boundary at $t=0$. 
Assuming that $\bu(t)$ is given and is independent of $\bB$, the induction equation~\eqref{eqn:induct} is of standard parabolic form (like the heat equation), so its solution is smooth for all $t>0$. If the initial condition $\bB(0)$ is also smooth and satisfies the boundary condition~\eqref{eqn:bc}, then the solution is smooth also at $t=0$, except possibly at $r=1$. For the solution to be smooth everywhere, including at $r=1$, and for Taylor's substitution to be valid, we need the initial condition to satisfy not only the usual boundary condition (also termed the zero order boundary conditions) but also the first order boundary conditions: that $\partial_t \bB$, given by ${\mathcal I}(\bB,\bu)$ of \eqref{eqn:induct} satisfies the boundary condition \eqref{eqn:bc}. Higher-order variants of the boundary conditions pertain to higher-order time derivatives. Assuming that this analysis extends to the full magnetostrophic equations, it provides strong constraints on the form of the initial condition that produces a solution that is smooth for $t \ge 0$ and all $r \ge 0$.


This issue of lack of smoothness of $B_\phi$ occurs only instantaneously at $t=0$. One may ask if it is possible to specify an initial field that satisfies Taylor's constraint and higher order boundary conditions, making it possible to use \cref{eqn:2ODE} directly. Although in principle the answer is yes, it would be practically impossible because an evaluation of the first order boundary condition requires knowledge of $\partial_t \bB$ and therefore $u_g$. The logic is therefore circular: we need to know $u_g$ in order to check the method that enables us to find $u_g$ in the first place. 
It would seem that some additional insight or good fortune would be required to find a geostrophic flow that is self-consistently satisfies the boundary conditions.
The complication compounds the already difficult task of finding an initial condition that satisfies the necessary condition of being a Taylor state.

It is worth noting, However, that once the system has evolved past the initial condition many of these problems vanish.
For $t>0$, solutions to parabolic systems are smooth and so automatically satisfy all higher order boundary conditions. It follows that \cref{eqn:2ODE} is valid for $t>0$, although this does not help find the geostrophic flow at $t=0$. 

 
\subsection{Schemes in which the boundary information is included}
These concerns described above regarding boundary conditions do not carry over to the axisymmetric case, the plane layer situation nor the 3D implicit schemes described.
In the axisymmetric and Cartesian cases \citep[e.g.][]{abdel1988alphaomega}, the boundary conditions are evaluated to zero and the boundary value of the magnetic field or any of its time derivatives never enter any subsequent calculations. In the 3D implicit scheme, because of the representation of all quantities (including $\bB$ and any of its time derivatives) in terms of a Galerkin basis, boundary conditions to all orders are satisfied.

Thus in the axisymmetric and Cartesian cases, \cref{eqn:Wuinstant} and \cref{eqn:1ODE} are correct irrespective of the initial choice of Taylor state, as is the fully implicit method of \cref{sec:3D_implicit} for the 3D case. This is to be contrasted with \eqref{eqn:2ODE} that is valid only for the subset of Taylor states satisfying zero and first order boundary conditions. 

\section{An appraisal of Taylor's method}
\label{sec:Taylor_appraisal}
\subsection{An illustration of when Taylor's method fails} \label{sec:Tayfail}

We are now in a position to provide a first explicit demonstration that Taylor's ODE \cref{eqn:2ODE} fails when using an initial Taylor state that does not satisfy first order boundary conditions. We show this in two parts. Firstly, within axisymmetry, we demonstrate that Taylor's \cref{eqn:2ODE} is formally inconsistent with the BWR equation \eqref{eqn:Wuinstant}; secondly, we plot an explicit solution of Taylor's equation and show that does not agree with those derived from other methods known to be correct. In sections 8--10 we will show that our generalised version of Taylor's method shows agreement among all methods.

We consider the simple case of the dipolar, single spherical harmonic $l=1$ axisymmetric poloidal magnetic field

\[ \bB = \curl \curl Ar^2 (30 r^4-57r^2+25)\cos(\theta) \hat{\bf r}, \]
where $A=\sqrt{231/20584}$ is a scaling constant (see \cref{sec:scaling}). We note that $\bB$ satisfies the electrically insulating boundary conditions \eqref{eqn:bc}, and is an exact Taylor state owing to its simple symmetry.


The ageostrophic flow (determined for example by the method described in Appendix \ref{sec:u_a_method}) has only an azimuthal component given by
\begin{eqnarray} u_\phi &= A^2 \big[9120s^7+(50400z^2-26184)s^5+(50400z^4-95760z^2+23888)s^3+ \nonumber \\  & (16800z^6-47880z^4+42000z^2-6824)s \big]. \end{eqnarray}

From \cref{eqn:Wuinstant}, the geostrophic flow satisfies 
\begin{equation} \label{eqn:manipWuinstant} \frac{d}{ds}\left(\frac{u_g(s)}{s} \right)=-\frac{S_0(s)}{s\alpha_0(s)} = -\frac{Q_5(s^2)}{sQ_2(s^2)}
\end{equation}
where we have used
\begin{align} \alpha_0(s) = \alpha(s)&=\frac{198}{2573} s^4 \pi (1-s^2)^{3/2} (640s^4-1168s^2+535), \nonumber \\
S_0(s) = S(s)&=-\frac{66528}{86064277}s^4\pi(1-s^2)^{5/2}(46387200s^8-138624000s^6+142265512s^4-57599212s^2+7255185),\nonumber \end{align}
and where, for typographic purposes, $Q_N(s^2)$ is used for brevity to represent a polynomial of order $N$ in $s^2$.
Substituting this into Taylor's equation \eqref{eqn:2ODE}, along with 
$$ \beta(s) =-\frac{198}{2573}s^3\pi(1-s^2)^{1/2}(7680s^6-17440s^4+12456s^2-2689),$$
%
%
%
leaves an unbalanced equation: the left and right hand sides of \eqref{eqn:2ODE} are the distinct quantities
$$\frac{\sqrt{1-s^2} Q_9(s^2)}{Q_2(s^2)}, \qquad \sqrt{1-s^2}Q_7(s^2).$$
%
%
%
Therefore, for this choice of $\bB$, none of the solutions of the first order ODE \cref{eqn:Wuinstant} satisfy the second order ODE \cref{eqn:2ODE}.
Full equations are given in the supplementary material.

This specific case (which is illustrative of the general case) shows that \cref{eqn:2ODE} and \cref{eqn:Wuinstant} are inconsistent: in particular the first order \cref{eqn:Wuinstant} is not simply the first integral of the second order \cref{eqn:2ODE}. The reason why they are not consistent is that although the ODEs are derived from the equivalent forms \eqref{eqn:Taylor_alternative} and \eqref{eqn:Taylor}, the boundary terms are used to derive \eqref{eqn:Wuinstant} but not \eqref{eqn:2ODE}. Thus the two equations embody different information. 
In this example, Taylor's method is equivalent to the erroneous replacement of $\partial_t B_\phi$ (which is zero) in the boundary term of \eqref{eqn:alter1}, by ${\mathcal I}_\phi \ne 0$. 
%
There appears to be no simple way of amending the coefficients $\alpha$ and $\beta$ for the derivation of \eqref{eqn:2ODE} to include the boundary information.
While the initial magnetic field has been chosen such that it satisfies the boundary condition \eqref{eqn:bc}, through computing $\partial_t \bB$ we can show that, based on Taylor's solution, the initial rate of change of the magnetic field violates this boundary condition.  

To confirm that Taylor's method is not generally valid, we now directly compare solutions from various methods.
Integrating \cref{eqn:Wuinstant} analytically gives the solution 
\begin{eqnarray} u_g = & \frac{A^2 s}{918060} \bigg[ 9926860800\,{s}^{6}-32213813760\,{s}^{4} +37855940880\,{s}^{2}+C -11143964160\,\ln  s + \nonumber \\ &  30664844\, \sqrt {21}\arctan \left(  \left( 80\,{s}^{2}-73 \right)/\sqrt{21} \right)   + 101695629\,\ln  \left( 640\,{s}^{4}-1168\,{s}^{2}+535 \right)  \bigg]. \label{eqn:axipol_u_g} \end{eqnarray}
We note that the solution is a sum of odd polynomials, an $s \ln(s)$ term and additional (and non singular) $\ln$ and arctan terms. The constant $C$ is determined through enforcing zero solid body rotation (equation \eqref{eqn:angmomcon}).
The solution for $u_g$ is everywhere continuous and finite, only at $s=0$ is there a weak singularity: $\partial_s (u_g/s) \sim 1/s$.
We also observe that there is no singularity at $s=1$.
A comparable analytic solution but for a quadrupolar axisymmetric magnetic field was given in \citet{Li_etal_2018}, which is also regular everywhere except for a weak $s \ln(s)$ singularity at $s=0$.
\begin{figure}
\centering
\includegraphics[width=0.7\textwidth]{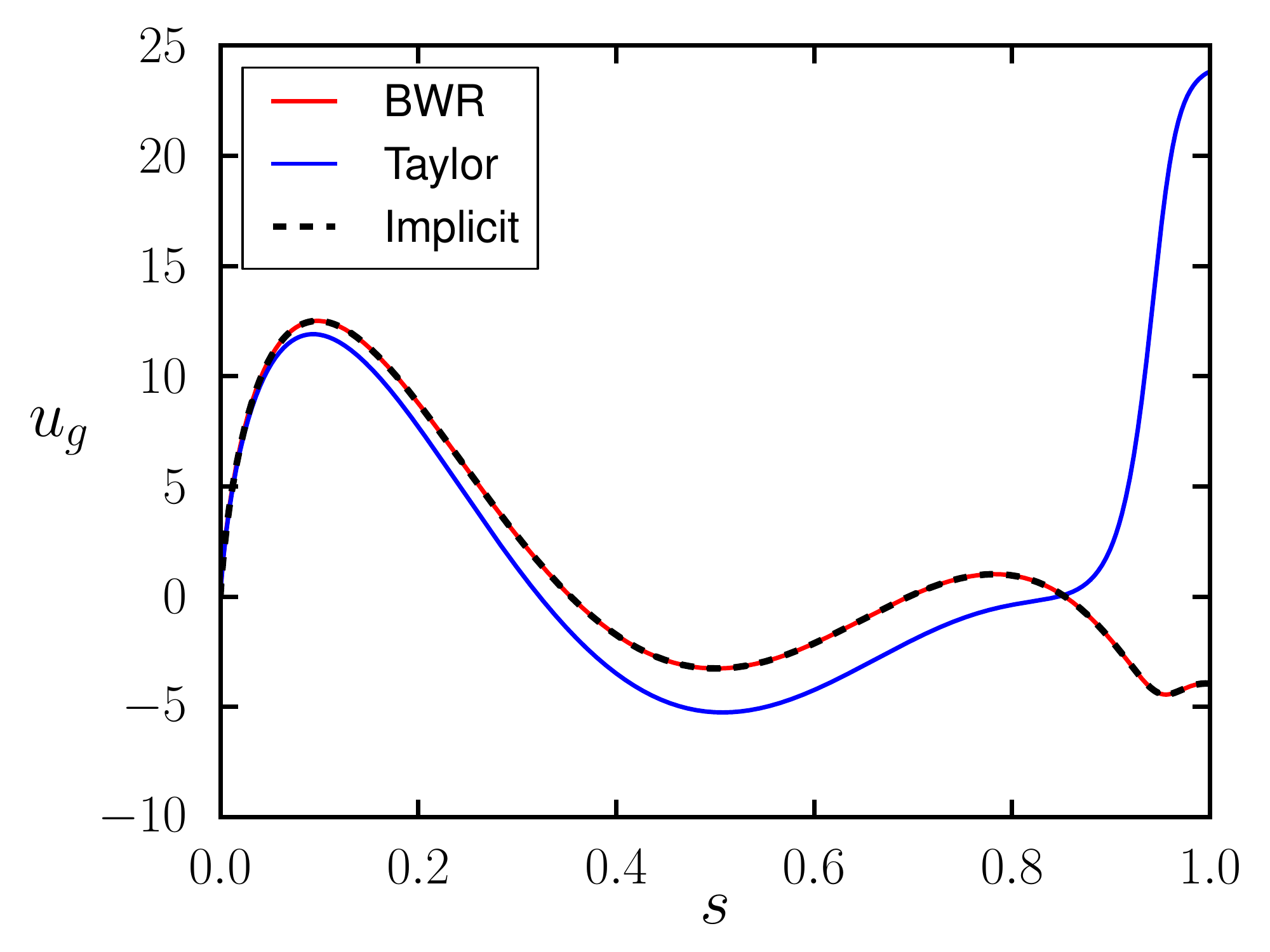}
\caption{ \label{fig:Example1_axisymmetricpoloidal}  Comparison of solutions for the geostrophic flow for an axisymmetric dipolar poloidal initial field. Red is the analytic solution of the first order BWR \cref{eqn:Wuinstant}, blue is a numerical solution of Taylor's second order ODE (see text) and dashed black is the solution using the implicit time step method with $h=10^{-9}$.}
\end{figure}
That the analytic expression \eqref{eqn:axipol_u_g} is indeed the true solution is confirmed by figure \ref{fig:Example1_axisymmetricpoloidal} which compares it to 
the geostrophic flow given by the independent 3D implicit scheme of section \ref{sec:3D_implicit}; the two solutions over-plot. 
A contour plot of the total azimuthal flow is shown in \cref{sec:Earth-like} (\cref{fig:contour_u_phi_polaxi}).

We now directly compare this solution with that obtained by solving Taylor's equation \eqref{eqn:2ODE}, shown as the blue line of \cref{fig:Example1_axisymmetricpoloidal}.
This solution is found by adopting the expansion \eqref{eqn:specexp} and minimising the integrated squared residual
\begin{equation} \int_0^1 \big[ {\cal L}_T(u_g) - G(s) \big]^2 ds.\label{eqn:Taylor_method} \end{equation}
with respect to the spectral coefficients, whose truncation is increased until the solution converges. 

Although all solutions agree at small $s$, Taylor's solution shows significant differences from the others for $s>0.8$.

\begin{figure}
\centering
\includegraphics[width=0.7\textwidth]{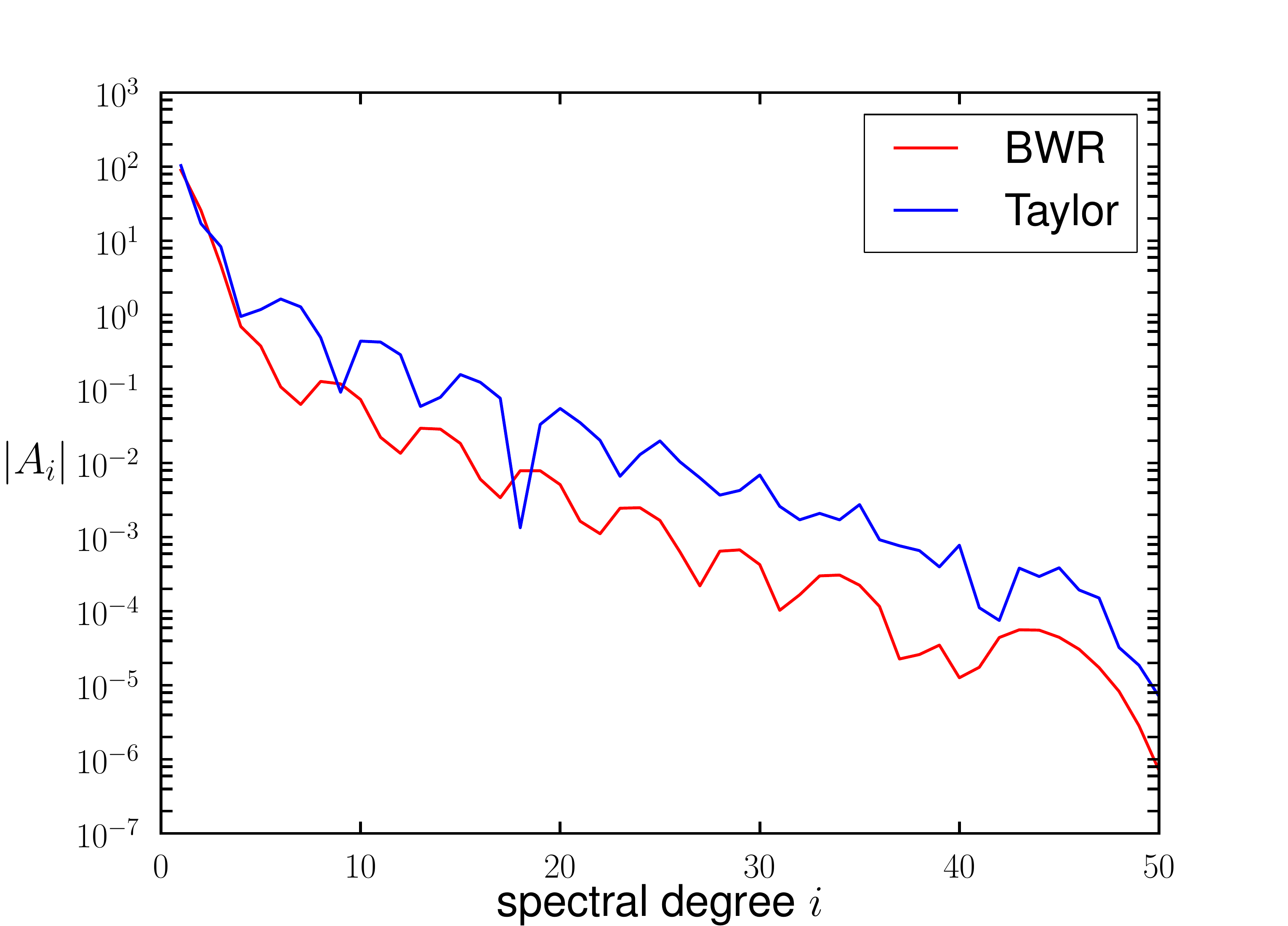}
\caption{A comparison of the absolute value of the polynomial spectral coefficients $A_i$, defined in \cref{eqn:specexp}, against degree for numerical solutions using the Braginsky-Wu-Roberts and Taylor formulations.} \label{fig:spect_pow_log_50_abs}
\end{figure}

It is also of interest to assess numerical convergence to solutions of \cref{eqn:2ODE,eqn:Wuinstant}.
Although we have an analytic solution to \eqref{eqn:Wuinstant}, we use the same numerical method as given above but now applied to \eqref{eqn:Wuinstant} by minimising
\begin{equation} \int_0^1 \big[ {\cal L}_{BWR}(u_g) + S_0(s) \big]^2 ds. \end{equation} 
\Cref{fig:spect_pow_log_50_abs} demonstrates that convergence of the solution is faster for the correct, first order equation \eqref{eqn:Wuinstant} than for Taylor's equation \eqref{eqn:2ODE}. 
Therefore, aside from Taylor's equation being generally inapplicable, it seems that converged solutions are also relatively more difficult to find. 
%
%
\subsection{Specific cases when Taylor's method succeeds}\label{sec:Tay_right}

For arbitrary purely-toroidal Taylor states bounded by an electrical insulator, $\bB$ vanishes on $r=1$ and in this special case Taylor's methodology is correct. This is because the boundary term involving $B_\phi B_r$ (see \cref{eqn:Taylor_alternative}) has a ``double zero'' and so, when considering its time derivative, erroneous substitution for $\partial_t {\bB}$ leaves it invariant as zero.

Taking the time derivative of \eqref{eqn:Taylor_alternative}, noting that the boundary term is zero, we obtain
\begin{equation} s^2 \int_{C(s)} \left( \frac{\partial B_\phi}{\partial t} B_s  + B_\phi \frac{\partial B_s}{\partial t} \right) d\phi dz = 0. \label{eqn:tor1} \end{equation}
Using the 3D extension of \eqref{eqn:simpinduct} 
\begin{equation} \label{eqn:3D_induct} \curl(u_g(s) \phihat  \times\bB) = s B_s \frac{\partial (u_g/s)}{\partial s} \phihat - \frac{u_g}{s} \frac{ \partial_1 {\bB}}{\partial \phi}
\end{equation}
where $\partial_1/\partial \phi$ is a derivative with respect to $\phi$ that leaves invariant the unit vectors \citep[see e.g.][]{Book_Jault_2003}, 
the term involving $u_g$ in \eqref{eqn:tor1} becomes
$$ s \frac{d}{ds}\left(\frac{u_g(s)}{s} \right) \int_{C(s)} B_s^2   d\phi dz - \frac{u_g}{s} \int_{C(s)} \frac{\partial}{\partial \phi} \left( B_\phi B_s \right) d\phi dz. $$
Noting that the last integral is zero, we obtain an equation (that holds in 3D) that is of the same form as the axisymmetric BWR equation \eqref{eqn:Wuinstant} 
\begin{equation} s \alpha(s)   \frac{d}{ds}\left(\frac{u_g(s)}{s} \right) = - S(s). \label{eqn:purtor} \end{equation}

As an illustration we consider the non-axisymmetric $l=1$, $m=1$ toroidal magnetic field
\[ \bB = \curl A r^2 (1-r^2)\cos(\phi)\sin(\theta) \hat{\bf r}, \]
where $A$ is a scaling constant which takes the value $\frac{3}{4}\sqrt{105}$.
The ageostrophic flow is
\begin{align}
\vec{u}_{a}&=\frac{A^2}{3}s\sin\phi\cos\phi(5s^4-6s^2z^2-3z^2-3z^4-10s^2+6z^2+5) \shat  \nonumber \\ &+\frac{A^2}{15}(\cos^2\phi (105s^5-30z^2s^3-130s^3-15z^4s+30z^2s+25s)-56s^5+72s^3-16s)\phihat \nonumber \\
&+\frac{4A^2}{3} s^2z(3s^2+z^2-3)\cos\phi\sin\phi \zhat,
\end{align}
and, solving \eqref{eqn:purtor}, the geostrophic flow is
\begin{equation} \label{eqn:nosing}
u_g(s)=A^2\left(\frac{97}{30}s^5-\frac{77}{15}s^3+sC_1\right),
\end{equation}
where $C_1$ is determined through considerations of angular momentum. Note the  absence of singularities in this solution.

This geostrophic flow is shown in \cref{fig:Example2_nonaxisymmetrictoroidal_alt}, and we note that the 3D implicit method and Taylor's method give the same solution (not shown).

\begin{figure}
\centering
\includegraphics[width=0.7\textwidth]{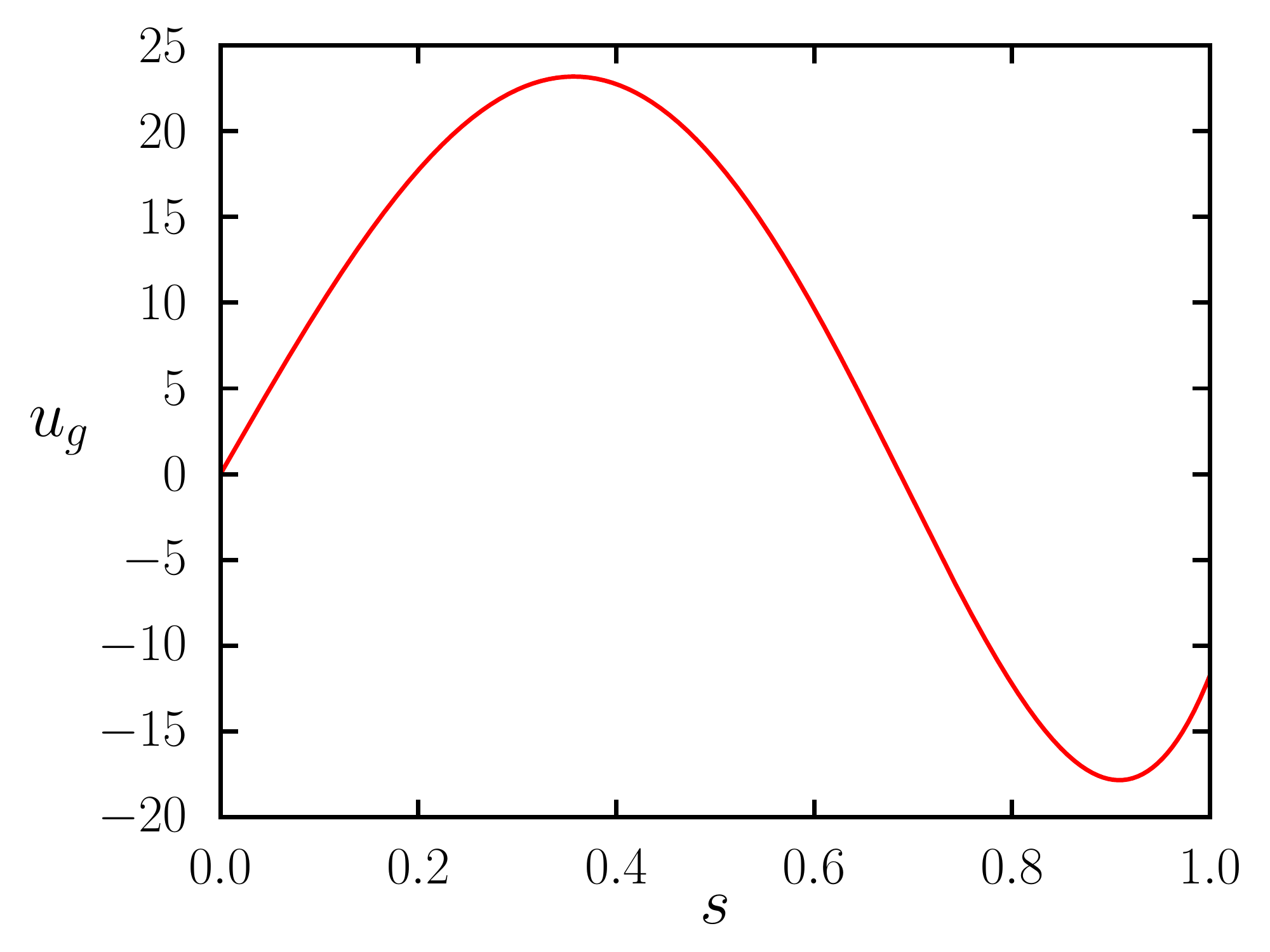}
\caption{The geostrophic flow corresponding to a non-axisymmetric $l=1$, $m=1$ purely-toroidal Taylor state. 
} \label{fig:Example2_nonaxisymmetrictoroidal_alt}
\end{figure}

We can show analytically that for non-axisymmetric toroidal fields, Taylor's equation \eqref{eqn:2ODE} and \eqref{eqn:purtor} are equivalent, up to the requirement of a further boundary condition for the second order differential \cref{eqn:2ODE}. We do this by taking $s^{-1} \partial_s$ of  \cref{eqn:purtor} to obtain 
$$\frac{1}{s} \frac{d}{d s} \bigg[ s\,\alpha(s) \frac{d}{ds} \left(\frac{u_g(s)}{s} \right) \bigg]=- \frac{1}{s} \frac{d S(s)}{d s} = G(s)$$
where the right hand equality is due to the equivalence of \cref{eqn:manipulation1,eqn:Taylor}. By comparing coefficients of derivatives of $u_g$ we reproduce \cref{eqn:2ODE} if and only if $$\beta(s)=\frac{\alpha(s)}{s}+\frac{d\alpha (s)}{ds}.$$ 

As shown by \cite{Roberts_Aurnou_2011} (their B.26c), exploiting $\grad \cdot \vec{B} = 0$ and the vanishing of a toroidal field on $r=1$, in this specific case $\beta$ can be written as
$$\beta(s)= \int_{C(s)} \left(3B_s^2 + s\, \pd{(B_s^2)}{s}\right)\, s\, d\phi\, dz .$$

We can then write $$\frac{\alpha(s)}{s}+\frac{d\alpha(s)}{ds}  = \int_{C(s)} s B_s^2\, d\phi\, dz + \frac{d}{ds} \int_{C(s)} s^2 B_s^2\, d\phi\, dz = \int_{C(s)} \left(3B_s^2 + s\, \pd{(B_s^2)}{s}\right)\, s\, d\phi\, dz$$
where the right most equality makes use of Leibniz's integral rule  with the boundary condition $\vec{B}={\bf 0}$ on $r=1$.
%
Hence, in this specific case, both ODEs are equivalent and therefore Taylor's ODE is valid.

\section{A generalisation of Taylor's analysis} 
\label{sec:Altmeth}

To modify the method of Taylor so that it applies to a magnetic field that does not satisfy the first order boundary conditions, 
we use \eqref{eqn:alter1} to impose stationarity of the Taylor constraint. 
Equally, we could impose stationarity of the equivalent equation \eqref{eqn:Taylor_alternative} but it is simpler to avoid the additional integral in $s$.
Bearing in mind
%
our discussion in \cref{sec:cty}, we take particular care to ensure correct handling of the boundary term. 

The magnetic field matches continuously (since  $\eta \ne 0$) with
an external potential field within the mantle $r \ge 1$. Note that our assumption of a globally continuous solution differs from the case when $\eta = 0$, for which horizontal components of $\bB$ may be discontinuous on $r=1$ \citep{Book_Backus_etal_96}.
In our setting where $\eta \ne 0$, the field matches continuously but not necessarily smoothly across $r=1$.
We note however that owing to $\nabla \cdot\bB = 0$, the radial component of $\bB$ (and all its time derivatives) are always smooth at $r=1$ \citep[see e.g.][]{Gubbins_Roberts_87}: thus only the horizontal components $B_\theta$ and $B_\phi$ are not in general smooth.

Thus, in the first term of \cref{eqn:alter1} we may substitute at $t=0$ 
\begin{align} 
&\partial_t B_s = {\mathcal I}_s(\bu,\bB), \qquad 0 \le r < 1,\nonumber \\
&\partial_t B_\phi = {\mathcal I}_\phi(\bu,\bB), \qquad 0 \le r < 1. \label{eqn:Bdot_ind} \end{align}
For the second (boundary) term, we may substitute
$\partial_t B_r = {\mathcal I}_r(\bu,\bB)$
but the initial value of $\partial_t B_\phi$ at $r=1$ is not specified by ${\mathcal I}_\phi$ alone, as assumed in Taylor's derivation. 

The key remaining issue is then to find 
the initial boundary value of $\dot{B}_\phi$, for which we present three  methods below.  Having done this, all terms are defined and \eqref{eqn:alter1} provides an implicit determination of $u_g$ up to the usual considerations of solid body rotation and regularity.

We observe that the form of \cref{eqn:alter1} differs markedly from \cref{eqn:2ODE,eqn:Wuinstant}: in addition to the spatial derivatives of $u_g$ (in the leftmost term), there is an explicit boundary term. For the general case, this boundary term must be retained, although it may be neglected under certain circumstances: e.g. those of sections  \ref{sec:Tay_right} and \ref{sec:Earth-like}.

We remark that the above instantaneous method can be amended to a first order implicit scheme (akin to \cref{eqn:1ODE}) by considering 
\begin{equation} \frac{1}{s} \frac{\partial}{\partial s} \bigg[ s^2 \int_{C(s)} (\dot{B_\phi} B_s + B_\phi \dot{B_s}) d\phi dz \bigg] + \frac{s}{\sqrt{1-s^2}} \oint_{N+S} (\dot{B}_\phi B_r + B_\phi \dot{B_r}) d\phi=- \frac{1}{hs}\pd{\Gamma_z(s,t)}{s},\label{eqn:genimplicit} \end{equation}
As before, this equation is applicable even when $\Gamma_z \ne 0$, that is, if the solution is close but not exactly on the Taylor manifold.

\subsection{A potential-based spherical transform method}
\label{sec:potential-based}
One way to find $\dot{B}_\phi$ on $r=1$ is to note that it is the azimuthal component of 
the potential field in $r\ge 1$
$$\dot{\bf B} = -\nabla \dot{V}, \;\; \nabla^2 \dot{V}=0. $$
The potential $\dot V$ is itself determined through continuity of the radial component $\dot{B}_r$ at $r=1$ and thus depends upon $u_g$. 
%
%
%
This method of determining $\dot{B_\phi}$ has been introduced in the study of torsional waves by \cite{Book_Jault_2003}, but is implemented here for the evaluation of the geostrophic flow.

The time derivative of the potential $\dot{V}$ can be written in terms of orthonormal spherical harmonics $Y_{lm}$ with unknown coefficients $a_{lm}$ as
\[ \dot{V} = \sum_{l,m} {\dot a}_{lm} r^{-(l+1)} Y_{lm}, \]
where $0 \le l \le L_{max}$ and $-l \le m \le l$ and
\[ \dot{a}_{lm} = \frac{1}{l+1} \oint_{r=1} \dot{B_r} Y_{lm} d\Omega, \]
where $\Omega$ is an element of solid angle. It follows then that on $r=1$
\[ \dot{B}_\phi = -\frac{1}{\sin\theta} \sum_{l,m} \dot{a}_{lm} \frac {\partial Y_{lm}}{\partial \phi}. \]

Key to the implementation of this method here is a spectral expansion of $u_g$, for example \eqref{eqn:specexp}, because it allows $\dot{B}_r$ (which depends on the $I+2$ spectral coefficients of $u_g$) to be evaluated everywhere on the boundary, as required in the above spherical transform. This is to be contrasted for example with a finite difference representation of $u_g$ where no such evaluation is possible. 

To find $u_g$, we note that all time-derivative terms in the left hand side of \eqref{eqn:alter1}, including those evaluated on the boundary, are linear in the unknown coefficients $(A_0,A_1,\dots,A_I,B)$, and hence the residual is of the form
$$ R(s) = \sum_{i=0}^I A_i a_i(s) + B b(s) + c(s)$$
for some functions $a_i$, $b$ and $c$ that depend on $\bB$ and $\bu_a$.
We formulate a single equation for the coefficients defining $u_g$ by minimising the quantity $\int_0^1 R^2 ds$ (which is quadratic in the coefficients that we seek). Note that the solution is approximate and depends on two parameters $I$ and $L_{max}$, which represent the truncation of the expansion used and care must be taken to ensure we achieve convergence in each.

\subsection{A potential-based Green's function method}

An alternative method for determining the potential $\dot{V}$ at the core mantle boundary is through the use of a Green's function convolved with ${\dot B}_r$ on $r=1$. Following \cite{gubbins1983use, Johnson_Constable_97}, the relevant Green's function associated with the Laplace equation in the exterior of a sphere with Neumann boundary conditions is
$$N(x,\mu) =\frac{1}{4\pi}\left( \ln \left(\frac{f + x - \mu}{1 - \mu}\right) - \frac{2x}{f}\right), $$
where $x=\frac{1}{r}$, $f=(1-2x\mu+x^2)^\frac{1}{2}$, $\mu=\cos\theta\cos\theta'+\sin\theta\sin\theta'\cos(\phi-\phi').$
This can be expressed as $N(x,\mu)=N\left(\frac{1}{r},\theta,\theta',\phi - \phi' \right)$, which is the potential at location $(r,\theta, \phi)$ in $r\ge 1$ due to a singularity of unit strength in the radial field at $(\theta', \phi')$ on the core-mantle boundary.
%
Making use of the periodicity of $\phi$, the magnetic potential in the region $r \ge 1$ can then be written as

$${\dot V}=\int_0^{2\pi}\int_0^\pi {\dot B}_r(1,\theta',\phi-\phi')N\left(\frac{1}{r},\theta,\theta',\phi \right)\sin\theta' d\theta' d\phi',$$
and so
%
$${\dot B}_\phi(1,\theta,\phi)=-\frac{1}{r\sin\theta}\int_0^{2\pi}\int_0^\pi \pd{{\dot B}_r(1,\theta',\phi-\phi')}{\phi}N\left(\frac{1}{r},\theta,\theta',\phi' \right)\sin\theta' d\theta' d\phi'.$$
Like the previous method, this procedure of evaluating ${\dot B}_\phi$ on $r=1$ requires an integral over all solid angle. Using again our spectral expansion \eqref{eqn:specexp} this results in ${\dot B}_\phi$ being a linear function of the unknown spectral coefficients; thus using \cref{eqn:alter1} the geostrophic flow can then be determined as in section \ref{sec:potential-based}.
%
\subsection{A modal projection}\label{sec:protheo}
A further alternative method to find $\dot{B}_\phi$ on $r=1$, which does not rely on a magnetic potential, is to employ a modal basis set for the magnetic field that is complete and satisfies the required boundary conditions. Here we adopt a numerically expedient Galerkin basis set (see \Cref{sec:Apb} for details), whose orthonormal poloidal and toroidal modes are written respectively as $\vec{\mathcal{S}}_{(l,n)}^m$ and $\vec{\mathcal{T}}_{(l,n)}^m$. 

By using such a representation, boundary conditions to all orders are automatically satisfied and therefore a direct substitution of the projected representation of ${\mathcal I}$,
\begin{equation} \label{eqn:projinduct} 
\proj{ {\mathcal I}} = \sum_{l,m,n} c_{l,m,n} \vec{\mathcal{S}}_{(l,n)}^m + d_{l,m,n} \vec{\mathcal{T}}_{(l,n)}^m, \end{equation}
for $\partial_t \bB$ in all three components for the whole sphere $r \le 1$ is justified. 
In the above, $l$ is bounded by $L_{max}$, $0 \leq n \leq N_{max}$ and $\proj{\bf x}$ indicates the modal projection of $\bf x$ (see Appendix \ref{sec:projmeth}).

As before, key to the method here is the spectral representation \eqref{eqn:specexp} for $u_g$; the coefficients $c_{l,m,n}$ and $d_{l,m,n}$, found by integration (see \cref{sec:projmeth}), then depend linearly on the unknown coefficients $A_i$ and $B$.



%
\Cref{eqn:alter1} can be then written as the following, in which $u_g$ appears explicitly

\begin{equation} \frac{1}{s} \frac{d}{d s} \bigg[ s\,\alpha(s) \frac{d}{ds} \left(\frac{u_g(s)}{s} \right) \bigg] + \frac{s}{\sqrt{1-s^2}} \oint_{N+S} \bigg[ B_\phi \big\{ {\curl(u_g \phihat  \times \bB)} \}_r  + B_r  \big\{\proj{ \curl(u_g \phihat  \times \bB) }\}_\phi \bigg] \, d\phi= \tilde{G}(s) \label{eqn:corrected_T1} \end{equation}
and
\begin{eqnarray} -\tilde{G}(s)  =  \frac{1}{s} \frac{\partial}{\partial s} \bigg[ s^2 \int_{C(s)} ({C}^{a}_\phi B_s + {C}^{a}_s \, B_\phi) d\phi dz \bigg] + \frac{s}{\sqrt{1-s^2}} \oint_{N+S} {B_\phi} ({C}^{a}_r + {B_r} [\proj{ \vec{C}^{a}}]_\phi ) d\phi. \label{eqn:correction1} \end{eqnarray}


Although on one level a simpler method than those previously presented because we do not need to calculate $\dot{V}$, in fact the method is more computationally expensive for two reasons. First, we need to check convergence in three parameters: $I, L_{max}, N_{max}$, rather than just the first two; second, because the orthonormality requires an integration over radius, in addition to the integration over solid angle required by both methods. 

\section{Examples of the geostrophic flow in 3D}
\label{sec:3D_examples}

We now give some examples to illustrate our generalised methodology for computing the instantaneous geostrophic flow associated with 3D Taylor states, using our spherical-transform method. These will be compared with the solution obtained using the fully implicit 3D method with a very small timestep of $h = 10^{-9}$; in all cases the solutions overplot. In none of the cases is an analytic solution available for comparison. For further comparison we plot also the solution of Taylor's ODE (see \cref{eqn:Taylor_method}).


We consider firstly an example of a non-axisymmetric $l=2$, $m=2$ poloidal magnetic field 
%

\begin{equation} \bB = \curl \curl A\frac{45\sqrt{3}}{4}r^3(7-5r^2)\sin^2\theta \cos 2\phi\, \hat{\vec{r}} \label{eqn:purepol} \end{equation}
where $A = 1/({6\sqrt{390}})$. \Cref{fig:Example3_nonaxisymmetricpoloidal} shows that the implicit and instantaneous solutions agree, whereas similarly to the axisymmetric case of \cref{fig:Example1_axisymmetricpoloidal} we can see that Taylor's solution differs significantly particularly near $s=1$.

For all our three-dimensional solutions the expansion for $u_g$ differs from that in axisymmetry given in \cref{eqn:specexp}. We now don't include a logarithmic term. As discussed in \cref{sec:sing_s0}, the logarithmic behaviour is not expected outside of axisymmetry and would violate the assumed regularity of the magnetic field.

The approximate polynomial solution, with coefficients rounded to 5 significant figures, is
$$u_g=-94.079s+550.14s^3-2196.4s^5+3292.7s^7-2178.4s^9+11996s^{11}-35435s^{13}+42961s^{15}-24113s^{17}+5248.3s^{19},$$
where the expansion has been truncated at $s^{19}$ and convergence achieved with parameters $I = L_{max} = 20$.



 \begin{figure}
 \centering
 \includegraphics[width=0.7\textwidth]{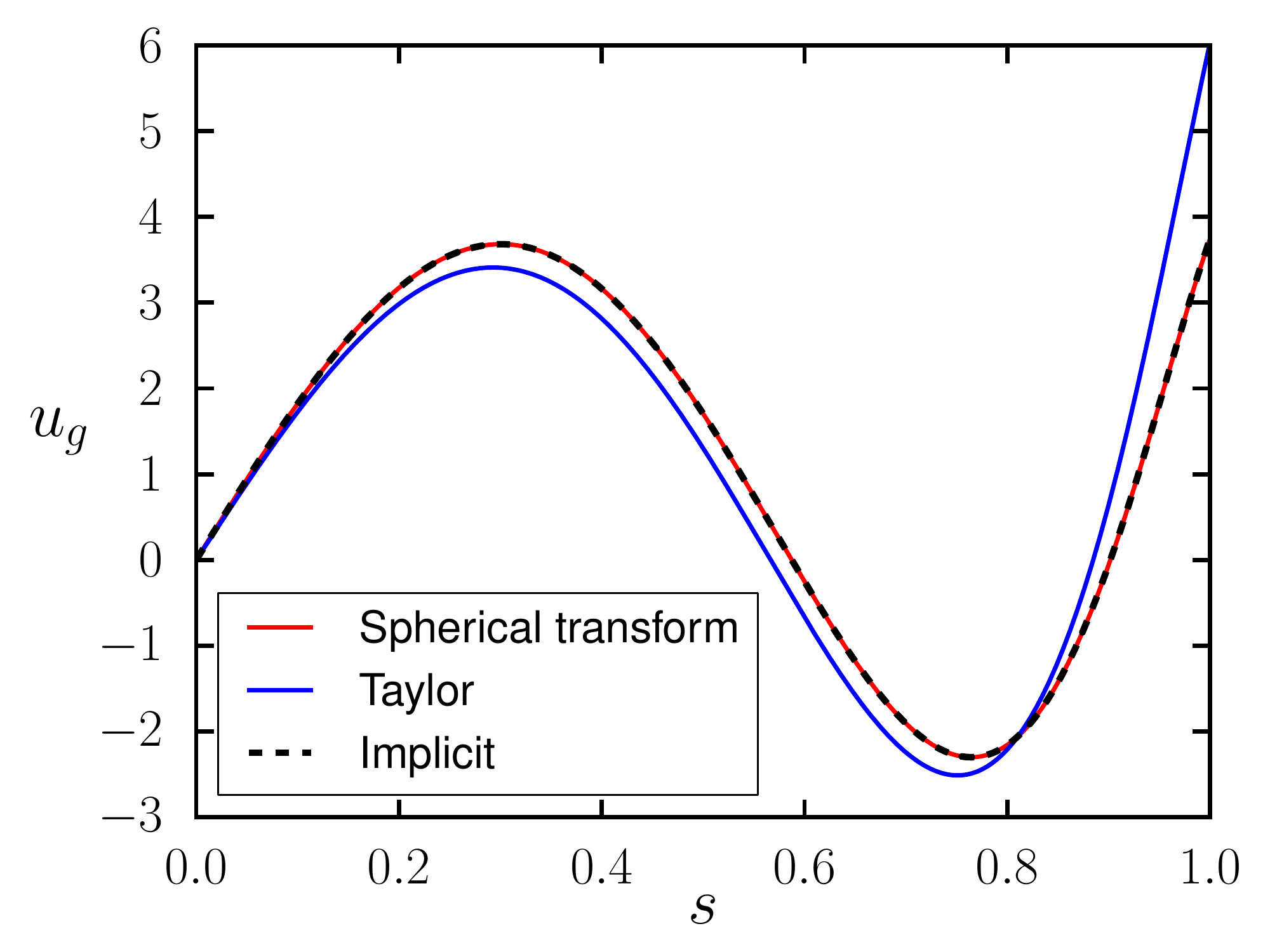}
 \caption{The geostrophic flow for the non-axisymmetric $l=2$, $m=2$ poloidal Taylor state of equation \eqref{eqn:purepol}. Solutions using the spherical transform method, the implicit timestep method with $h = 10^{-9}$ and Taylor's ODE are compared.} \label{fig:Example3_nonaxisymmetricpoloidal}
 \end{figure}

We secondly consider a more complex example of a non-axisymmetric magnetic field, which contains both $l=1$, $m=1$ toroidal and poloidal components
%
\begin{equation} \bB =  \curl A_t \sqrt{3} r^3 (1-r^2) \sin\theta \cos\theta \cos\phi\,\hat{\vec{r}} + \curl \curl A_p \frac{45\sqrt{3}}{2}r^3(7-5r^2)\sin\theta\cos\theta\cos\phi \,\hat{\vec{r}} \label{eqn:torpol} \end{equation}
where $A_t = \frac{5}{4}\sqrt{21}$ and $A_p = \sqrt{7/262440}$.
 \begin{figure}
 \centering
 \includegraphics[width=0.7\textwidth]{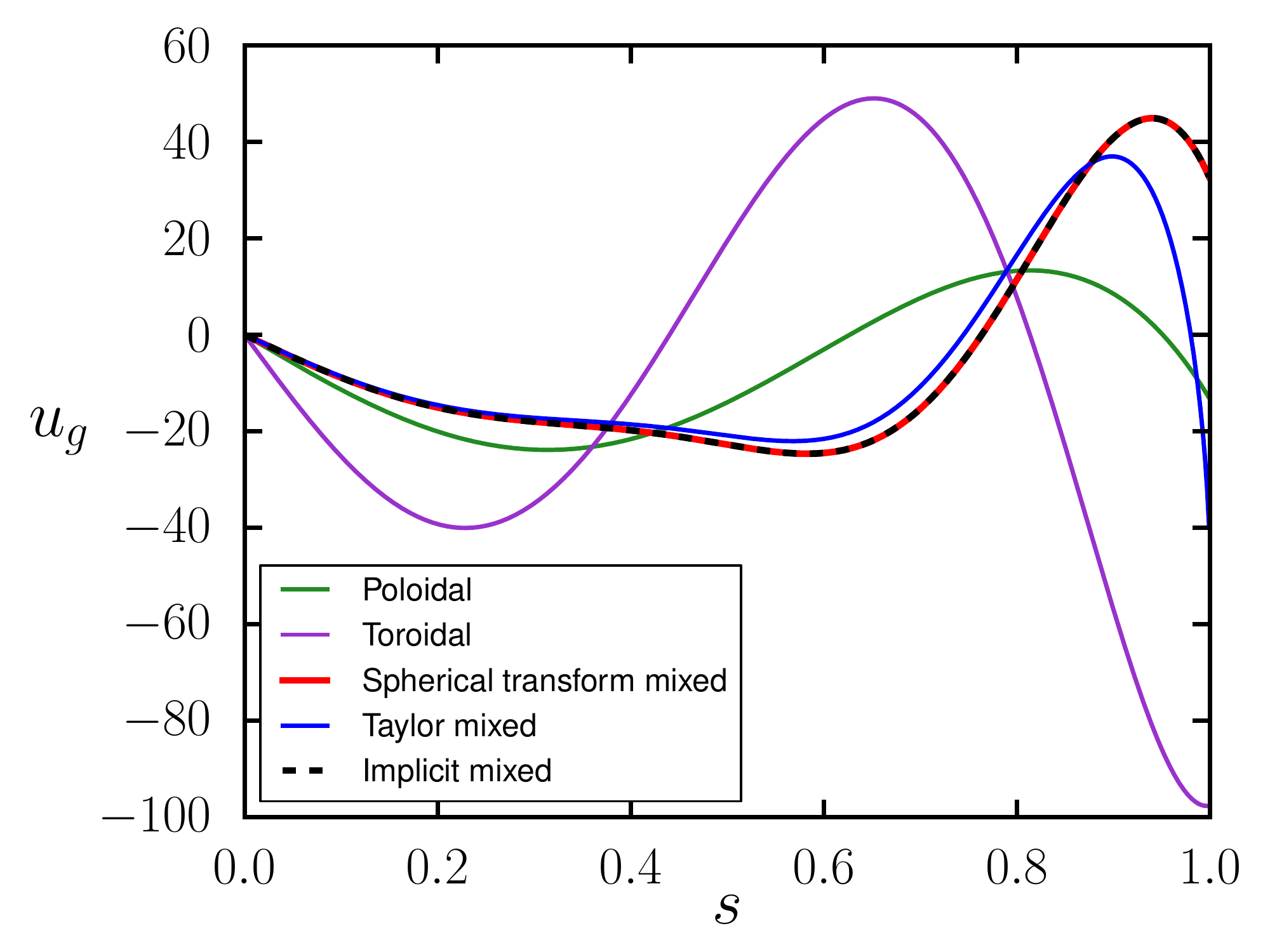}
 \caption{The geostrophic flow for the $l=1$, $m=1$ non-axisymmetric mixed Taylor state of equation \eqref{eqn:torpol}. Solutions using the spherical transform method, the implicit timestep method with $h = 10^{-9}$ and Taylor's ODE are compared. Solutions for solely either the poloidal and toroidal components of the Taylor state using the spherical transform method are also shown.} \label{fig:Example4_nonaxisymmetricmixed}
 \end{figure}
\Cref{fig:Example4_nonaxisymmetricmixed} shows that again the solution using the instantaneous method is validated by the implicit method, whereas Taylor's solution deviates as $s \to 1$. The figure also shows the geostrophic flow generated separately by either the purely-toroidal, or purely-poloidal magnetic field component, each individually a Taylor state. As anticipated by the structure of the equation for $u_g$ (nonlinear in $\bB$), the geostrophic flow driven by the total field does not equal the sum of the individually driven geostrophic flows.  

%
%
\section{Analytic approximation for an Earth-like field}
\label{sec:Earth-like}
%

Based on the present structure of the geomagnetic field, various studies show that it is reasonable to neglect the boundary term in \cref{eqn:Taylor_alternative} in an Earth-like context \citep{robertsWu2014,Roberts_King_2013}. 
This is because not only is the magnetic field likely much stronger inside the core than on $r=1$, but also because only the non-axisymmetric field contributes to the boundary term and it is relatively weak. The estimated strength of the magnetic inside the core is 5~mT, and that of the non-axisymmetric field on $r=1$ is 0.1~mT; therefore the relative magnitude of the boundary to the interior terms is $1/50^2\approx 0.04\%$. The negligible effect of the boundary term has been verified in the case of related studies of torsional waves \citep{jault2005alfven,Roberts_Aurnou_2011}.

Should we neglect the boundary term entirely, then the geostrophic flow is described by the same equation \eqref{eqn:purtor} that pertains to a purely-toroidal field, whose solution is

\begin{equation} u_g(s) = -s \int_0^s \frac{S(s')}{s' \alpha(s')}~ \text{d}s'. \label{eqn:Analapprox} \end{equation}
If $\alpha(s) > 0$ then this equation is integrable. A continuous solution for $u_g$ does not exist, however, if $B_s^2$ is everywhere zero on a geostrophic cylinder $C(s^*)$ (rendering $\alpha(s^*)=0$). Physically, this would mean that the magnetic field fails to couple cylinders on either side of $s=s^*$, leading to a discontinuity in the geostrophic flow.

In the Taylor states we use, $\bB$ is of polynomial form and it then follows that $S$ and $\alpha$ are also polynomial (up to a square root factor arising from the geometry) and therefore $u_g$ can (in general) be found in closed form. 
We note that, in general, $S/\alpha$ is $O(1)$ and so $u_g$ behaves as $s\ln(s)$ as $s \to 0$.






As an example of this approach, here we construct an Earth-like Taylor state comprising an axisymmetric poloidal mode and a non-axisymmetric toroidal mode, scaled such that the magnitude of the asymmetric part is $20\%$ of the magnitude of the axisymmetric part, but that the total rms field strength is unity:

\begin{equation} \bB =  \curl \big[A_t \frac{\sqrt{3}}{2}r^3(1-r^2) \sin^2\theta \cos 2\phi\big] \hat{\vec{r}} + \curl \curl \big[A_p \frac{21}{2}r^2(5-3r^2)\cos\theta\big] \hat{\vec{r}} \label{eqn:earthlike} \end{equation}
where $A_t=\sqrt{28875}/4$ and $A_p = 1/\sqrt{966}$.
%
%
%
%
The analytic solution of \eqref{eqn:Analapprox} is
%



 \begin{align} u_g(s) = &\frac {s}{1185586336} \Big[ -3645348420\,\sqrt {10626}
\arctan \left( {\frac { \left( 5\,{s}^{2}-5 \right) \sqrt {10626}}{42}
} \right) +
9801464537150\,{s}^{6}  \nonumber \\ 
& - 12073529601375\,{s}^{4}-633064443000\,{s}^{2}-25808428800\,\ln  \left( s
 \right)   \nonumber \\ 
& + 25531026444\,\ln  \left( 6325\,{s}^{4}-12650\,{s}^{2}+6367
 \right) +1185586336\,{C_1} \Big]
 , \nonumber
 \end{align}
which is shown in 
\cref{fig:Exampleearthlike} and compared to our solution by the method in \cref{sec:potential-based} in which full account is taken of the boundary terms.
As anticipated, the two solutions are very similar and diverge only close to $s=1$ (where the boundary term has most effect), with an rms difference of about $1\%$, all of which occurs very close to the outer boundary. This validates the neglect of the boundary term for this example, and 
indicates the significance of \cref{eqn:Analapprox} which can be used with confidence to analytically approximate the geostrophic flow generated by an Earth-like field.
However, we note the presence of a logarithmic singularity that (in view of an earlier comment) that we do not expect in a non-axisymmetric case; this is discussed in the following section.
\begin{figure}
\centering
\includegraphics[width=0.5\textwidth]{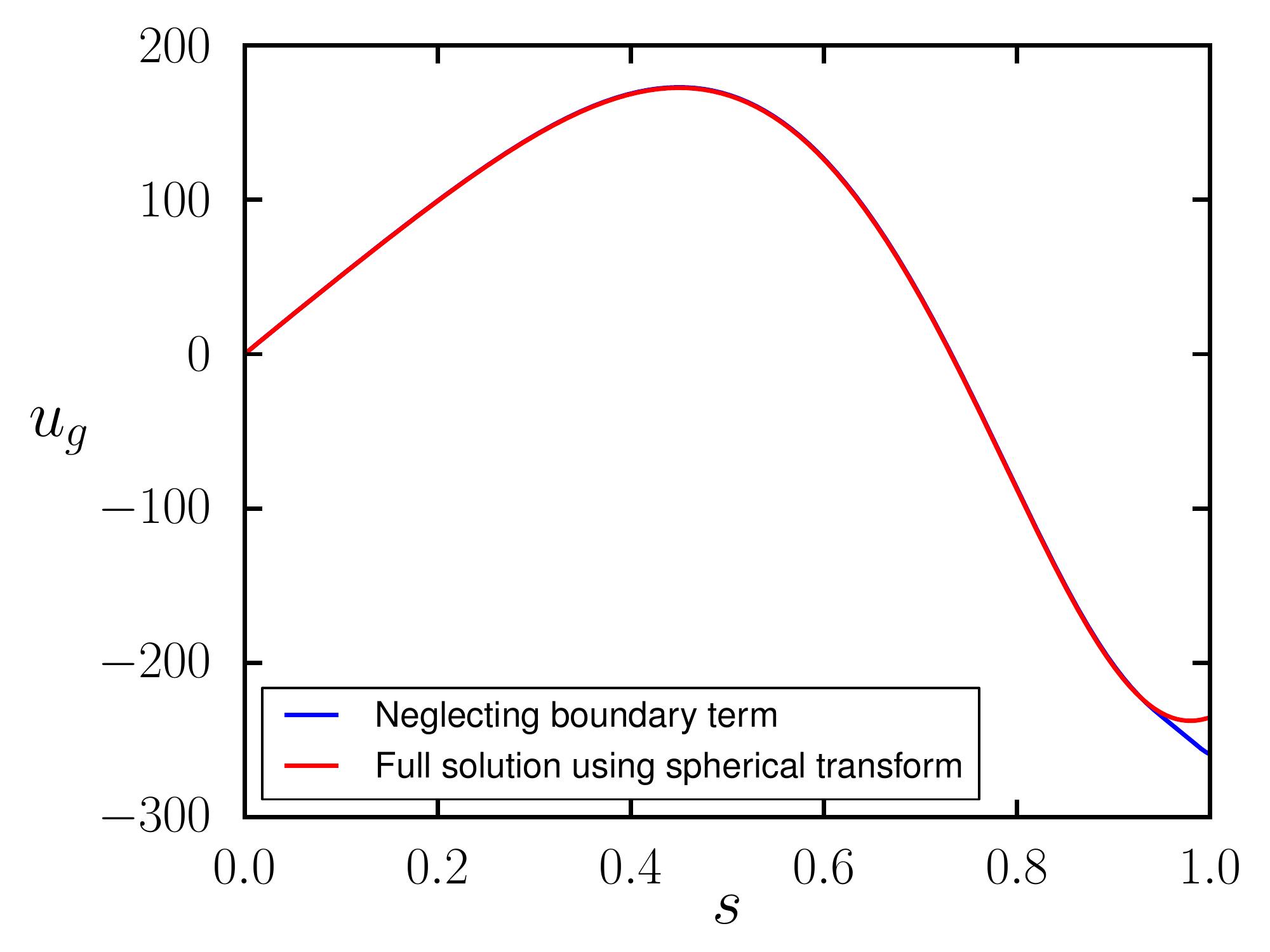}
\caption{The geostrophic flow for a non-axisymmetric Earth-like Taylor state. Numerical solution using the spherical transform method (red) is compared to the analytic solution neglecting the boundary term (blue).} \label{fig:Exampleearthlike}
\end{figure} 


Finally figure \ref{fig:contour_u_phi}(b) shows contours of the total azimuthal component of the flow. Of note is the much higher amplitude of flow associated with the increased complexity of the magnetic field compared to the single-mode magnetic field example of figure \ref{fig:contour_u_phi}(a). The scale of this flow is as would be expected geophysically: maximum dimensionless velocities are of order 100, corresponding to dimensional velocities of order $10^{-4}\text{ms}^{-1}$ consistent with large-scale core flows inferred by secular variation \citep{Holme_2015}.

\begin{figure} 
	\centering
	\begin{subfigure}{.49\textwidth}
		\centering
\includegraphics[width=0.7\textwidth]{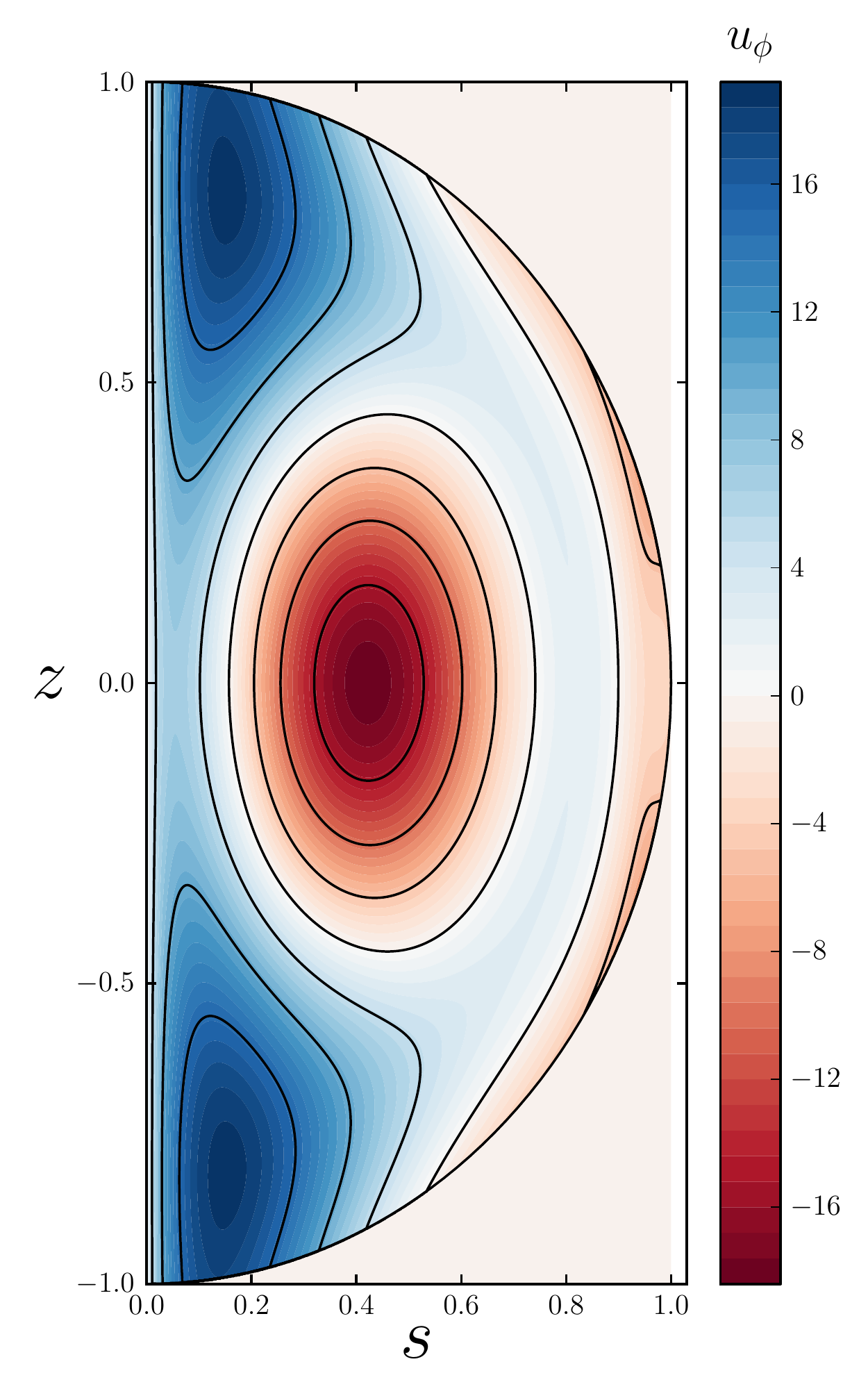}
\caption{}\label{fig:contour_u_phi_polaxi}
	\end{subfigure}%
	\begin{subfigure}{0.49\textwidth}
		\includegraphics[width=0.7\textwidth]{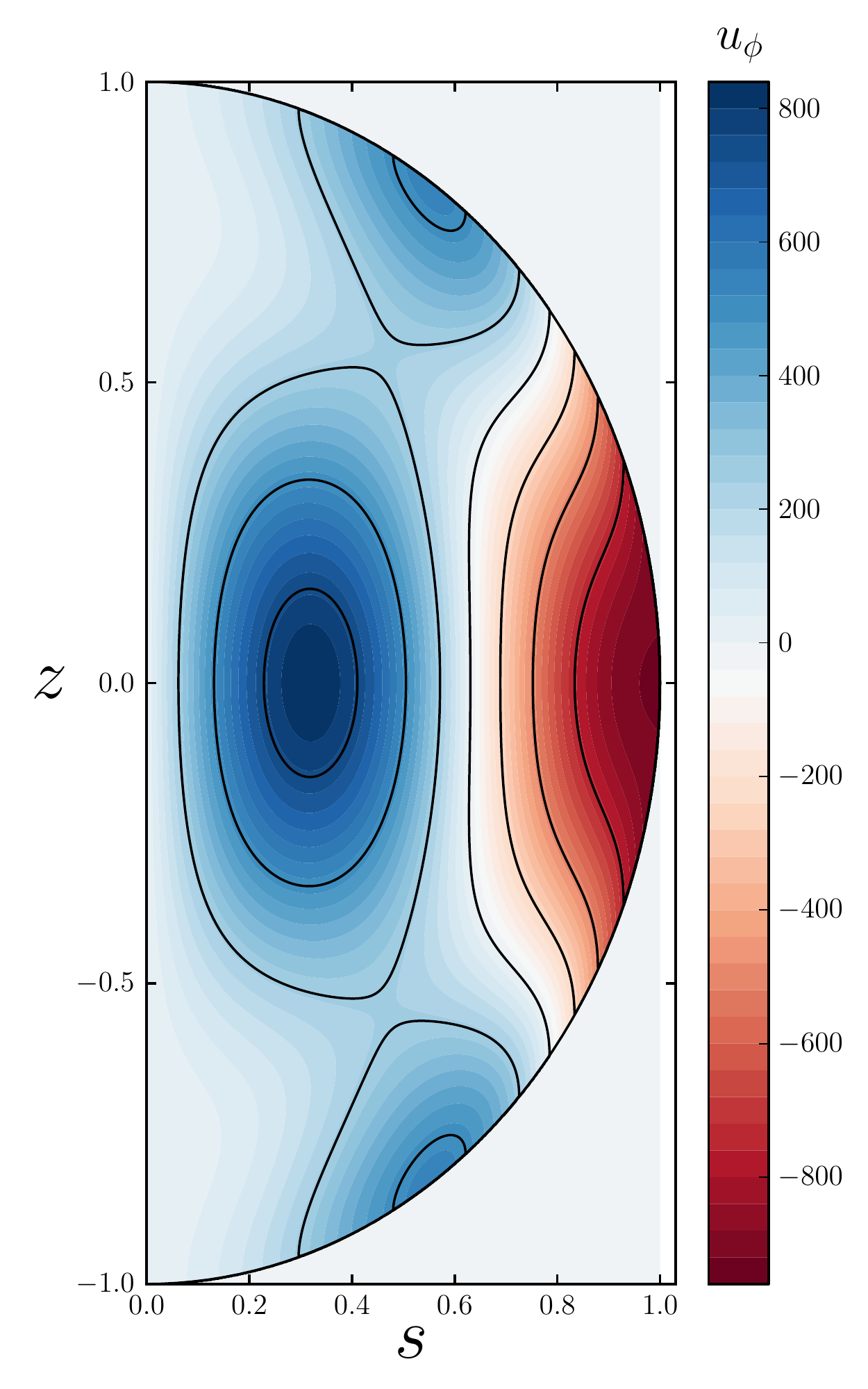}
\caption{}\label{fig:contour_u_phi_earthlike}
	\end{subfigure}
    \caption{Contour plots of (a) the total azimuthal flow $u_\phi$ driven by the axisymmetric poloidal field in \cref{sec:Tayfail}, (b) the axisymmetric part of the total azimuthal flow driven by the Earth-like field of \eqref{eqn:earthlike}. \label{fig:contour_u_phi}}
\end{figure}

\section{Singularities of $u_g$} 
\label{sec:sing}
A key benefit of having an instantaneous description of the geostrophic flow is to make explicit its analytic structure, which then motivates spectral expansions such as \eqref{eqn:specexp} for use with other methods. Assuming $\alpha(s) > 0$, because the equation describing $u_g$ is smooth and regular, $u_g$ is expected to be an odd \citep{Lewis_Bellan_90} finite function on $0<s<1$. There are three places however where the solution may be singular: (i) $s=0$; (ii) $s=1$ and (iii) in the complex plane $s = x+iy$, away from the real axis ($y \ne 0$). We discuss each in turn.

\subsection{Singularities at $s=0$} \label{sec:sing_s0}
Firstly we consider the presence of a singularity at $s=0$.
In axisymmetry, it is well established that $u_g \sim s \ln(s)$ as $s \to 0$, resulting in a $s^{-1}$ singularity in $\partial_s (u_g/s)$ \citep{Jault_95,Wu_Roberts_2015,Fearn_Proctor_87a}, reproduced in our example \eqref{eqn:axipol_u_g}.  However it has not been quite clear whether the logarithmic singularity pertains to a general asymmetric Taylor state: in particular, in axisymmetry $s=0$ is a singular line of the coordinate system, whereas in 3D spherical coordinates the only singular point is the origin $r=0$. 
\citet{robertsWu2014} showed that either by neglecting the boundary term (their (25a)) or considering Taylor's ODE directly, which we have shown to be of limited validity, (see their Appendix B) leads to a general logarithmic behaviour. 

At first inspection it appears that the boundary term is negligible as $s \to 0$. For a general 3D field, both $\bB$ and $\dot{\bB}$ are $\mathcal{O}(1)$ on $s=0$, suggesting that the interior term in \cref{eqn:alter1} is $\mathcal{O}(1)$, whereas the boundary term is $\mathcal{O}(s)$ as $s \rightarrow 0$. Motivated by the example in \cref{sec:Earth-like}, this suggests that a full treatment (including the boundary term) retains the singularity in 3D --- however, we do not find this to be the case. Significant cancellation in the interior term occurs and while the integrand is $O(1)$, the integral itself is $O(s)$, as expected since we know that the interior term and boundary term must sum to zero for all $s$. 
Therefore, there is no evidence that the 3D case has a logarithmic singularity at $s=0$, and indeed all our numerical solutions and analytic solutions are regular there. In the purely toroidal field explored in \cref{sec:Tay_right}, the analytic solution given in \cref{eqn:nosing} is purely polynomial, with no singular behaviour at the origin. This assertion can be strengthened into a theorem.

\begin{theorem} 
The assumption of a magnetic field that is regular initially and remains so for all time, places a restriction on the permitted behaviour of the geostrophic flow. 
In axisymmetry, the space of solutions allows a weak singularity in the geostrophic flow at $s=0$. However, in three dimensions it is required that the geostrophic flow is regular at the origin in order to maintain regularity of the magnetic field.

\end{theorem}

\begin{proof}

This result directly follows from the form of the geostrophic term in the induction equation. In axisymmetry this is given by \cref{eqn:simpinduct}, from which it is clear that it is permissible for $u_g$ to contain a weak logarithmic singularity while maintaining a regular $\bB$. In 3D the geostrophic term in the induction equation is given by \cref{eqn:3D_induct}. In the presence of a non-axisymmetric  magnetic field, any logarithmic singularity in $u_g$ would render $\partial_t {\bB}$ non-regular. 
Hence the assumption of regular $\bB(t)$ is incompatible with such a singular solution. 

\end{proof}



While the analytic approximation in \cref{sec:Earth-like} is shown to produce accurate geostrophic flows for Earth-like magnetic fields, it should be used with caution, since the analytic structure of the solution will contain an $s \ln s$ dependence, that does not persist when the full balance including the boundary term is considered.
For axisymmetric magnetic fields this weak logarithmic singularity is not a significant concern since the geostrophic flow only enters the induction equation through $\partial_s (u_g/s)$ and so the magnetic field remains regular everywhere. By contrast, in 3D the structure of the  geostrophic term in the induction equation (given in \cref{eqn:3D_induct}) means that the logarithmic singularity is imparted to the magnetic field itself, causing the magnetic field to diverge at the rotation axis and violating the standard assumption of a regular field. Thus, in a practical implementation, such singular behaviour must be filtered out of $u_g$.

\subsection{Singularities at $s=1$}
We also address the possible existence of a singularity at $s=1$. For the specific case of an axisymmetric dipolar magnetic field, \citet{robertsWu2014} presented an argument that $\partial_s u_g \sim (1-s^2)^{-1/2}$, although they conceded that this was not supported by their numerical examples. The same form of singular behaviour for $u_g$ has been predicted for torsional waves \citep{Schaeffer_2012,Marti_Jackson_2016}, perturbations to Taylor states, whose eventual steady state at $t=\infty$ would be exactly magnetostrophic (if indeed steady Taylor states exist). However, there is no reason why the analytic structure of the oscillations should mirror that of the underlying background state, particularly as the manner of how the limit $t \to \infty$ is reached at the end points where the wave speed may vanish is unclear \citep{Li_etal_2018}.  

Although we are not in a position to prove one way or the other the existence of singular behaviour at $s=1$, we demonstrate by example that it is not generally present.

We find no singularity at $s=1$ in the non-axisymmetric example of \cref{sec:3D_examples}. A similar regular behaviour is shown in \cref{fig:Tay_dugsds_sing_s1s} (red curve) for an axisymmetric example. Interestingly, for this latter case, the application of Taylor's ODE (which is invalid for this example) gives a solution that does show a singularity at $s=1$ (blue curve).
In this instance, singular behaviour is simply an artifact of applying Taylor's ODE when it is not valid, and we have found no cases where a solution to our more general analysis behaves singularly at $s=1$. 

This observation may help explain why the prediction of a singularity at $s=1$ \citep{robertsWu2014} is not borne out in any numerical examples. They themselves discussed this discrepancy and hypothesized that a key issue is the lack of boundary information contained within Taylor's equation. We speculate that should their magnetic field satisfy not only Taylor's constraint and the boundary conditions but also crucially the first order boundary conditions, that this singular behaviour will vanish and the geostrophic flow will remain regular at $s=1$. 
We note however that certain magnetic forcing terms can render the geostrophic flow singular at $s=1$: for example, that of a non-polynomial mean-field $\alpha$-effect described in Appendix F of \citet{Li_etal_2018}.

Finally, we remark that for a dipolar axisymmetric Taylor state,  \citet{Li_etal_2018} 
showed evidence of non-singular but abrupt boundary-layer like behaviour close to $s=1$, possibly because the equation describing the geostrophic flow is null at the equator (i.e. $\alpha = S = 0$). A similar result was also found by \citet{Fearn_Proctor_87a} who abandoned constraining their geostrophic flows near $s=1$ due to anomalous behaviour. We note, however, in our analytical solutions, we find no evidence of such behaviour: for example \cref{fig:Example1_axisymmetricpoloidal} shows a smooth solution at $s=1$.


\begin{figure} 
	\centering
	\begin{subfigure}{.49\textwidth}
		\centering
\includegraphics[width=\textwidth]{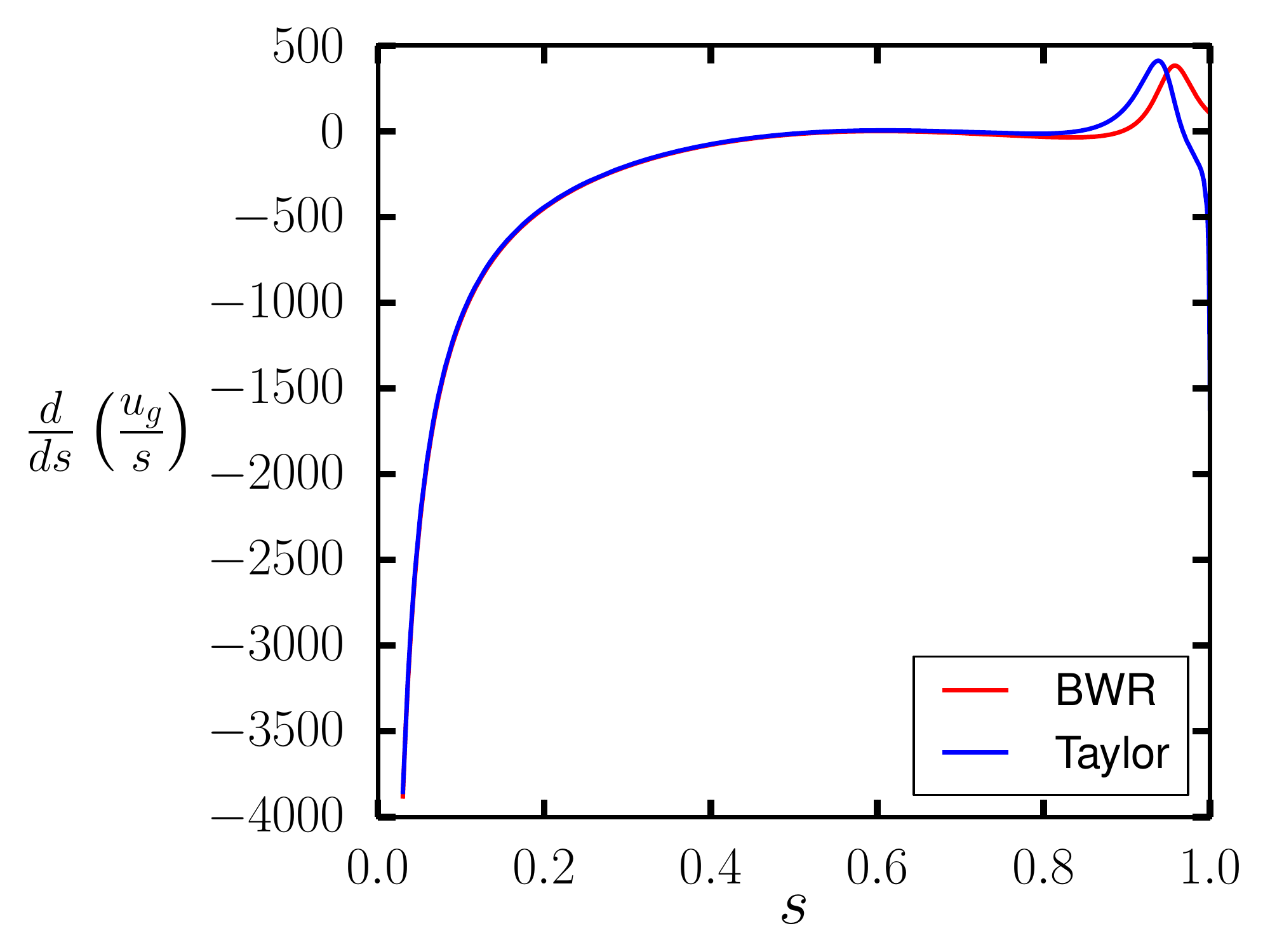}
\caption{}\label{fig:Tay_dugsds_sing_s1}
	\end{subfigure}%
	\begin{subfigure}{.49\textwidth}
		\includegraphics[width=\textwidth]{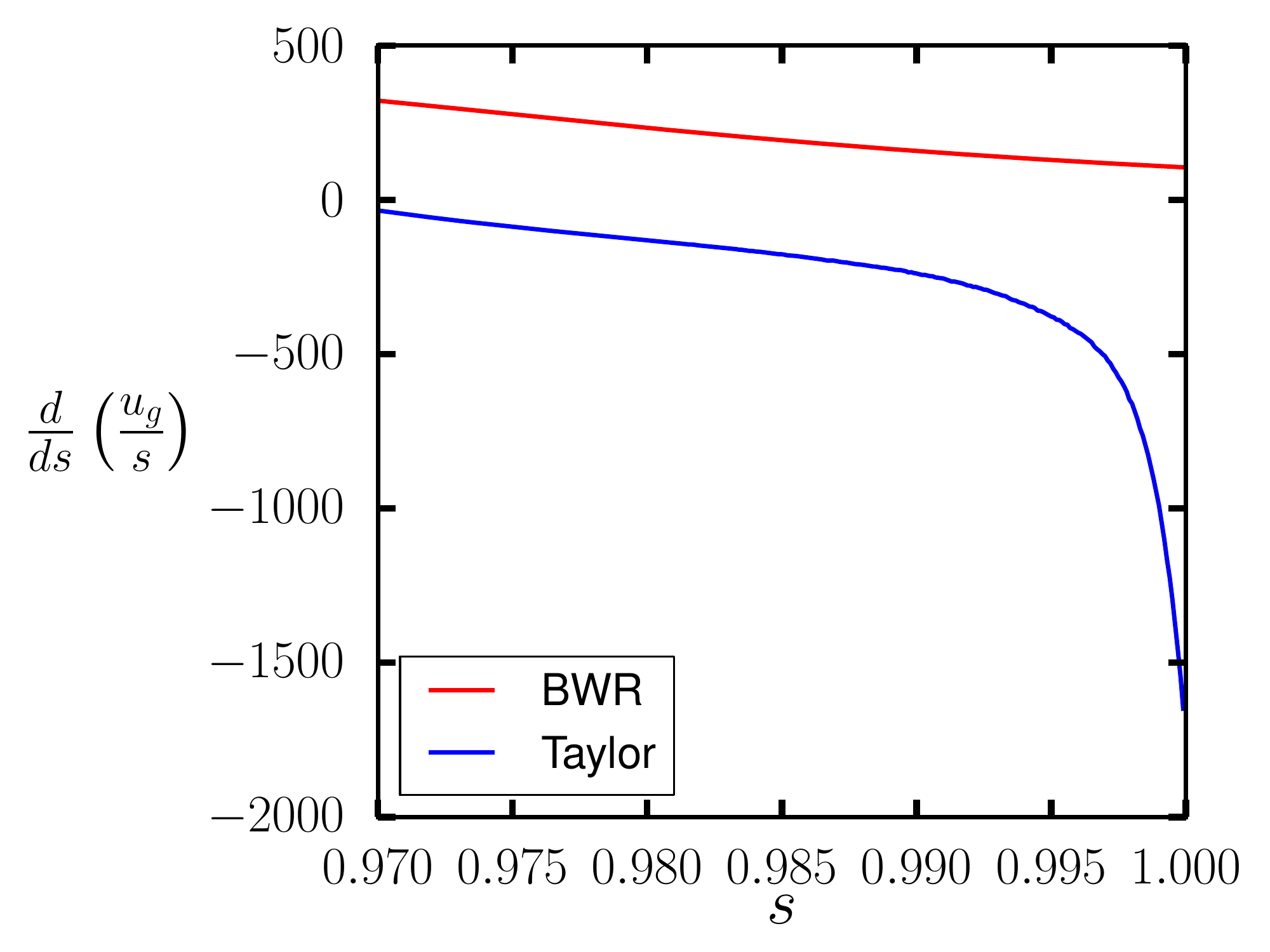}
\caption{}\label{fig:Tay_dugsds_sing_s1_zoom}
	\end{subfigure}
    \caption{A plot of $\partial_s (u_g/s)$ for solutions to a mixed axisymmetric Taylor state consisting of the poloidal field of the example of \cref{sec:Tayfail} with a $l=1$, $m=0$, $n=1$ toroidal Galerkin mode, using the BWR and Taylor equations. (a) Shows the whole domain, a singularity of the form $s^{-1}$ is visible for both solutions at $s=0$ and for Taylor's solution only, a weaker singularity also occurs at $s=1$. (b) Zoomed-in plot of the $s=1$ singularity to show clearly that it only occurs when solving Taylor's equation; it has the form $(1-s^2)^{-\frac{1}{2}}$.}\label{fig:Tay_dugsds_sing_s1s}
\end{figure}

\subsection{Singularities off the $s$-axis}
Finally, inspecting an example solution \eqref{eqn:axipol_u_g} shows that there can be either branch cuts or logarithmic singularities away from the real line. These do not affect the solution itself (defined on the real interval $0\le s \le 1$) but can influence convergence of the numerical method used to find $u_g$ \citep{Book_Boyd_2001}. The closer the singularities lie to the real interval $[0,1]$ the slower the convergence. In general, we speculate that such singularities can lie arbitrarily close to the real line, possibly being associated with the breakdown of the magnetostrophic balance, for example, torsional waves etc.

\section{Discussion}

In this paper we have discussed in some detail how the geostrophic flow, a fundamental part of any magnetostrophic dynamo, might be determined. Of particular note is that we have shown why the method introduced by \citet{Taylor_63} fails in most cases, because of its intrinsic (and, to date, unrecognised) assumption that the initial magnetic field structure must satisfy a higher-order boundary condition (that is, both the magnetic field and its time derivative must satisfy matching conditions pertaining to an exterior electrical insulator). We presented a generalised version of Taylor's method valid for an arbitrary initial magnetic Taylor state that is not subject to higher order boundary conditions. In many of our examples, the magnetic fields of dimensional scale 1.7~mT drive flows of magnitude about $10^{-4}\text{ms}^{-1}$, comparable to large-scale flows inferred for the core \citep{Holme_2015}. Thus, in concert with weakly-viscous models, inviscid models also produce Earth-like solutions.

A broader point of note is the extent to which the restriction on the validity of Taylor's approach impacts the related derivation of the equation describing torsional waves \citep{Roberts_Aurnou_2011}. A general treatment of torsional waves includes boundary terms, whose proper evaluation would require a method such as described in \citet{Book_Jault_2003}. However, the troublesome boundary terms are usually neglected, either because of axisymmetry or because of arguments based on the relative size of the asymmetric magnetic field \citep{Roberts_King_2013}. Either way, these approaches remain unconstrained by any consideration of higher order boundary conditions on the magnetic field and the theoretical description remains correct. However, in \cref{sec:sing_s0} we describe the danger of neglecting the boundary term, this leading to a logarithmic singularity not present in solutions of the full equation. This has potential implications for analysis of torsional waves, for which the avoidance of a logarithmic singularity may require the full boundary term.


It is worth noting that the weak logarithmic singularity $u_g \sim s \ln(s)$ as $s \to 0$ in axisymmetric magnetostrophic models stands in contrast with weakly viscous models which are anticipated to be regular everywhere. For example, the asymptotic structure is $u_g=O(s)$ in axisymmetry for both no-slip and stress-free boundary conditions (using the formulae summarised in equations 8 \& 9 of \citet{Livermore_etal_2016a} and the fact that $B_s, B_\phi \sim s$ as $s\to 0$) and $u_g=O(1)$ in non-axisymmetry (using the formulae in equation (33) of \cite{Hollerbach_96a} and the fact that $([\curl \bB] \times \bB)_\phi \sim s$ through the properties of general vectors described by \cite{Lewis_Bellan_90}). 
The presence of a weak logarithmic singularity is therefore a feature unique to the axisymmetric inviscid case, and serves to distinguish the exact magnetostrophic balance (with zero viscosity) from models with arbitrarily small but non-zero viscosity.
However, in 3D there is no such distinction between the structure of $u_g$ between $E=0$ and $E \ll 1$: in both cases $u_g$ is regular. 


Given that the geometry of the outer core of the Earth is a spherical shell rather than a full sphere, a natural question to ask is how would we calculate the flow within this domain. The method for determining the ageostrophic flow would remain comparable although it could be discontinuous or singular across the tangent cylinder $\cal C$, the geostrophic cylinder tangent to the solid inner core \citep{Livermore_Hollerbach_2012}. As for the geostrophic flow, in the absence of viscosity, there is no reason why it must be continuous across $\cal C$; there are no known matching conditions that it must satisfy and such an analysis lies far beyond the scope of this work.

Although supplying an analytic structure of the evolving magnetostrophic flow, an instantaneous determination of the geostrophic component is not itself of practical use within a numerical method using finite timesteps of size $h$, as the solution will immediately diverge from the solution manifold \citep{Livermore_etal_2011}. However, as for the  axisymmetric-specific method of \citet{Wu_Roberts_2015}, our 3D instantaneous methods generalise simply to schemes that are accurate to first order in $h$, thus presenting a viable method for numerically evolving a 3D magnetostrophic dynamo. A direct comparison of this method with the fully implicit (3D) method of \citet{Li_etal_2018} would be an interesting study. Indeed, our 3D first-order-accurate solutions could be used as a starting guess for their nonlinear iterative scheme, enabling much larger timesteps to be taken for which the geostrophic flow does not need to be close to its structure at the previous step.

Lastly, there is mounting evidence that rapid dynamics within the core is governed by quasi-geostrophic (QG) dynamics, in which the flow is quasi invariant along the axis of rotation \citep{pais2008quasi, Gillet_etal_2011}. We briefly comment on whether the slowly evolving background magnetostrophic state is also likely to show such a structure. 
Both \citet{Li_etal_2018} and \citet{Wu_Roberts_2015} show axisymmetric magnetostrophic solutions that have largely $z$-invariant zonal flows.
Here, in our 3D cases we also find that the geostrophic flow is comparable in magnitude to the ageostrophic zonal flow. In our Earth-like example, comparing figures \ref{fig:Exampleearthlike} and \cref{fig:contour_u_phi}(b), the maximum value of the geostrophic flow is about one quarter of that of the total zonal flow. 
Furthermore, our (large-scale) magnetostrophic solutions contain a significant $z$-invariant component, a finding that is consistent with
\cite{aurnou2017cross} who have recently suggested the existence of a threshold lengthscale, below which the geodynamo is magnetostrophic and above which the dynamics are QG. 

\section{Acknowledgements}

This work was supported by the Engineering and Physical Sciences Research Council (EPSRC) Centre for Doctoral Training in Fluid Dynamics at the University of Leeds under Grant No. EP/L01615X/1. The authors would also like to thank Andrew Jackson and Dominique Jault for helpful discussions, as well as Chris Jones, Rainer Hollerbach and David Gubbins  for useful comments.

\appendix

\section{Further details of numerical methods}
\subsection{A Galerkin representation} \label{sec:Apb}
A simple way of constructing magnetic states is to take combinations of single-mode toroidal or poloidal vectors, whose scalars are each defined in terms of a single spherical harmonic:
$$\vec{B} = \sum_{l,m,n} a_{l,n}^m {\vec{\mathcal{T}} }_{l,n}^m + b_{l,n}^m {\vec{\mathcal{S}}}_{l,n}^m  $$
where
$\vec{\mathcal{T}}_{l,n}^m=\curl (\chi_{l,n}(r) Y_l^m \rhat)$ and $\vec{\mathcal{S}}_{l,n}^m=\curl \curl (\psi_{l,n}(r) Y_l^m \rhat)$ and the harmonics are fully normalised over solid angle:
$$ \oint \big[ Y_l^m \big]^2 d\Omega = 1. $$
We choose the scalar functions $\chi_{l,n}$ and $\psi_{l,n}$, $n\ge 1$, to be of polynomial form \citep{li2010optimal,Li_etal_2011}, and defined in terms of Jacobi polynomials $P_n^{(\alpha,\beta)}(x)$, by

\begin{align} 
\chi_{l,n}&=r^{l+1}(1-r^2)P_{n-1}^{(2,l+1/2)}(2r^2-1) \nonumber \\
\psi_{l,n} &= r^{l+1} \bigg( c_0 P_n^{(0,l+1/2)}(2r^2-1) + c_1 P_{n-1}^{(0,l+1/2)}(2r^2-1)+c_2 \bigg) 
\end{align}
where
\begin{align}
c_0 &= -2n^2(l+1)-n(l+1)(2l-1)-l(2l-1) \nonumber \\
c_1 &= 2(l+1)n^2+(2l+3)(l+1)n+(2l+1)^2 \nonumber \\
c_2 &= 4nl+l(2l+1).
\end{align}
Suitably normalised, the vector modes then satisfy
(A) the boundary conditions of equation \eqref{eqn:bc}; (B) regularity at the origin and (C) $L_2$ orthonormality of the form
\begin{align} & \int_V {\vec{\mathcal{S}}}_{l,n}^m \cdot {\vec{\mathcal{S}}}_{l',n'}^{m'} \, dV = \int_V {\vec{\mathcal{T}}}_{l,n}^m \cdot {\vec{\mathcal{T}}}_{l',n'}^{m'} \, dV = \delta_{l,l'}\, \delta_{m,m'}\, \delta_{n,n'} \nonumber \\
& \int_V {\vec{\mathcal{S}}}_{l,n}^m \cdot {\vec{\mathcal{T}}}_{l',n'}^{m'} dV = 0, \label{eqn:Ap_orthonormal} \end{align}
where all integrals are over the spherical volume $V$. These conditions reduce to the equations (when $l = l'$, $m = m'$)
$$l(l+1) \int_0^1\frac{l(l+1)}{r^2}\psi_n \psi_{n'}+\pd{\psi_n}{r}\pd{\psi_{n'
}}{r} ~dr=\delta_{n,n'}, ~~ \text{and} ~~ l(l+1) \int_0^1\chi_n \chi_{n'} ~dr=\delta_{n,n'}.$$ 
For the velocity field, the ageostrophic flow satisfies only the impenetrable condition $u_r = 0$ on $r=1$, which constrains only the poloidal representation. A modal set that satisfies this boundary condition, regularity at the origin and $L_2$ orthonormality is given by \citet{Li_etal_2018}
$$ \vec{u} = \sum_{l,m,n} c_{l,n}^m {\bf t}_{l,n}^m + d_{l,n}^m {\bf s}_{l,n}^m  $$
where
$\vec{t}_{l,n}^m=\curl (\omega_{l,n}^m(r) Y_l^m \rhat)$ and $\vec{s}_{l,n}^m=\curl \curl (\xi_{l,n}^m(r) Y_l^m \rhat)$. The radial functions are given by 
\begin{align} \xi_{l,n}^m &= r^{l+1}(1-r^2)P_{n-1}^{(1,l+1/2)}(2r^2-1)  \nonumber \\
\omega_{l,n}^m &=r^{l+1}P_{n-1}^{(0,l+1/2)}(2r^2-1)\end{align}
for $n \ge 1$.

\subsection{Projection} \label{sec:projmeth}
In section \ref{sec:protheo} we need to project a divergence-free magnetic field $\bf B$ onto the magnetic Galerkin basis up to a truncation $L_{max}$ in spherical harmonic degree and $N_{max}$ in radial index:
$$ \proj{\bf B} = \sum_{l=1}^{L_{max}} \sum_{m=-l}^{l} \sum_{n=1}^{N_{max}} a_{l,n}^m {\bf T} + b_{l,n}^m {\bf S}_{l,n}^m.$$
Determination of the coefficients $a_{l,n}^m$ and $b_{l,n}^m$ can either be
accomplished through use of the 3D integral \eqref{eqn:Ap_orthonormal} directly, or equivalently by first taking the transform in solid angle to find the toroidal and poloidal parts of $\bf B$ 
\begin{equation} \label{eqn:tor}
T_l^m(r) = \frac{r^2}{l(l+1)} \oint (\curl \vec{B}_l^m)_r \,Y_l^m(\theta,\phi) ~ d\Omega, \qquad
S_l^m(r) = \frac{r^2}{l(l+1)} \oint (B_l^m)_r \,Y_l^m(\theta,\phi) ~ d\Omega,
\end{equation}
where $d\Omega = \sin\theta d\theta d\phi$, and secondly integrating in radius to give
$$a_{l,n}^m =\int_0^1 T_{l}^m \chi_{l,n} ~ dr ,\qquad b_{l,n}^m = \int_0^1\frac{l(l+1)}{r^2}S_{l,n}^m \psi_{l,n}+\pd{S_{l}^m}{r}\pd{\psi_{l,n}}{r} ~dr. $$

\subsection{Computation of the ageostrophic flow} \label{sec:u_a_method}


For a magnetic field $\vec{B}$ which is an exact Taylor state we can solve the magnetostrophic equation 
\begin{equation}\vec{\hat z} \times \vec{u} = -\grad p + \curl \vec{B} \times \vec{B}, \label{eqn:magneto_app}
\end{equation}
to determine the ageostrophic part of the fluid velocity $\vec{u}_{a}$. We note that the geostrophic flow is unconstrained by this equation as
\[ \vec{\hat z} \times u_g(s) \phihat= -u_g(s) \shat =  -\bn \int u_g(s)\,ds, \]
and so, as it can be written as a gradient, it can be absorbed into the pressure term.

The procedure then consists of taking the curl of \cref{eqn:magneto_app} to remove the pressure dependence and then proposing a trial form of the fluid velocity $\bu$ in terms of modes with unknown coefficients. 
Because $\vec{\hat z}$ is a constant vector and $\bB$ is based on Galerkin modes of polynomial form of known maximum degree, the modal representation for the flow then also has a known maximum degree. 
The unknown coefficients are found by equating powers of $r$ and solving the resulting system analytically with the assistance of computer algebra (e.g. Maple). 

It is worth noting that the solution $\bu$ above is determined only up to an arbitrary geostrophic flow. We remove the cylindrically-averaged azimuthal component of $\bu$, which results in the ageostrophic flow $\vec{u}_{a}$ with no geostrophic component. This also means that the geostrophic flow, determined through the methods described in the main text, is uniquely defined.

\bibliography{allrefs.bib}

\bibliographystyle{plainnat}

\end{document}